\documentclass[runningheads]{llncs}
\usepackage[T1]{fontenc}
\usepackage{cite}
\usepackage{amsmath,amssymb,amsfonts}

\usepackage{subcaption}

\usepackage{textcomp}
\usepackage{xcolor, enumerate,fullpage}

\usepackage[noend, linesnumbered, ruled, lined]{algorithm2e}

\SetAlFnt{\scriptsize}
\usepackage{comment}
\usepackage{enumitem}
\usepackage{lineno}
\usepackage{optidef}
\usepackage{graphicx}
\allowdisplaybreaks
\newtheorem{observ}{Observation}

\begin{document}
\title{Forest Covers and Bounded Forest Covers}
%
%
\author{Daya Ram Gaur\inst{1}\and
Barun Gorain\inst{2} \and
Shaswati Patra\inst{2}
\and Rishi Ranjan Singh\inst{2}}
\authorrunning{Gaur et al.}
%
\institute{University of Lethbridge, Alberta, Canada\\ 
\email{gaur@cs.uleth.ca}\\
 \and
Indian Institute of Technology Bhilai,  Durg,  Chattisgarh,  India\\
\email{\{barun,shaswatip,rishi\}@iitbhilai.ac.in}}
\maketitle              
\begin{abstract}
We study approximation algorithms for the forest cover and bounded forest cover problems. A probabilistic $2+\epsilon$ approximation algorithm for the forest cover problem is given using the method of dual fitting. A deterministic algorithm with a 2-approximation ratio that rounds the optimal solution to a linear program is given next. The 2-approximation for the forest cover is then used to give a 6-approximation for the bounded forest cover problem. The use of the probabilistic method to develop the $2+\epsilon$ approximation algorithm may be of independent interest.

\end{abstract}
\section{Introduction}

Let $G=(V, E)$ be an undirected graph with positive weights in the interval $[0,1]$ on the edges given by $w: E \rightarrow [0,1]$. A vertex cover $C \subseteq V$ is a subset of vertices such that every edge in $E$ is incident on some vertex in $C$. A minimum vertex cover is a vertex cover of minimum cardinality.  A tree is a connected component without any cycles. A collection of trees is called forest. The weighted index ($wi$) of a forest $F=\{T_1, T_2, \ldots T_k\}$ is calculated as follows: $wi(F)= \sum_{e\in E(F)} w_e+ k$, where $w_e$ is the weight of edge $e$, $k$ is the number of connected components in $F$ and $E(F)$ is the set of edges in $F$. The number of connected components $k = \sum_{i=1}^k |T_i| - |E(F)|$ where $|T_i|$ is the number of vertices in tree $T_i$. So, $wi(F) = \sum_{e\in E(F)} w_e+ \sum_{i=1}^k |T_i| - |E(F)|$ which can be rewritten as
$$ wi(F) = \sum_{i=1}^k |T_i| - \sum_{e \in E(F)} (1-w_e).$$ 

Our goal is to find a pair $(C, F)$ where $C \subseteq V$ is a vertex cover and $F \subseteq E$ is a forest in $C$ such that $wi(F)$ is the minimum possible. This problem is referred to as the forest cover problem. 

This is a generalization of the unweighted vertex cover problem. 
The unweighted vertex cover problem is one of the original 21 NP-complete problems \cite{karp2021reducibility}. Papadimitriou and Steiglitz \cite[p 432]{papadimitriou1998combinatorial} attibute a 2-approximation algorithm using maximal matching to Gavril and Yannakakis. Vertex cover in bounded degree graphs was studied by Berman and Fujito \cite{berman1995approximation} who have a $2 - 5/(d+3)+\epsilon$ algorithm for graphs with maximum degree $d$, and Hochbaum gave $2-2/d$ approximation algorithm for graphs with degree $d$ \cite{hochbaum1982approximation}. Vertex cover is hard to approximate within $2 -  (2+o_d(1)) \frac{\log\log d}{\log d}$ under an assumption known as the unique games conjecture \cite{austrin2009inapproximability}. This lower bound on the approximability matches the upper bound due to Halperin \cite{halperin2002improved} up to $o_d(1)$ factor. An unconditional lower bound of 1.36 on the approximation ratio is due to Dinur and Safra \cite{dinur2005hardness}. No approximation algorithm with approximation ratio $2 -\epsilon$ for a constant $\epsilon$ is known for unweighted vertex cover. Han, Punnen and Ye \cite{han2007polynomial} have $3/2 + \chi$ approximation algorithm for a parameter $\chi$; no examples where $\chi > 0$ were discovered in their extensive empirical evaluation. Parameterized algorithms for vertex cover have been studied for a while. Very recently, Harris and Narayanswamy \cite{harris2022faster} gave a $O^*(1.25284^k)$ algorithm to find a vertex cover of size $k$. This beats the previous long-standing bound of $O^*(1.2738^k)$ of Chen, Kanj and Xia \cite{chen2010improved} since 2010.

 Given an unweighted vertex cover instance, we can create an instance of the forest cover problem by assuming a weight function that assigns a weight of 1 to every edge. This is an approximation preserving reduction. Therefore, the problem is NP-complete, and the best approximation ratio that we can hope for is $2$ for the forest cover problem as an approximation ratio of $\alpha$ gives the same approximation ratio for vertex cover.

We study another related graph covering problem; the bounded forest cover problem (BFC). Given a graph $G=(V, E)$ with positive weights on the edges and a parameter $\lambda \ge 0$. The goal is to find a minimum-sized collection of trees $T_1, T_2, \ldots T_k$ such that the total weight of the edges in each tree $T_i$ is at most $lambda$ and the vertices in $\cup_{i=1}^k T_i$ form a vertex cover of the graph. By minimum-sized collection, we mean $k$ should be the smallest possible. We call this forest cover as opposed to tree cover because tree cover is already used with a different meaning.

The forest cover problem is inspired by graph covering problems and min-max vehicle routing problems \cite{Arkin93,Kim:1989,GPS23,EvenGKRS04,Guttmann:1997,Arkin2006}. It is also motivated by the need to dig tunnels or create crossing paths in mine-ridden areas \cite{Arkin93}. Similarly, the forest cover problem aims to establish tree-shaped facilities, considering a more general cost function with two components. One part evaluates the cost of travel, while the other accounts for the capital cost of deploying vehicles/robots, etc. The main goal is to minimize deployment and operating costs.

 The bounded forest cover problem~\cite{GPS23} addresses the gap between tree cover, minimum tree cover, and min-max tree cover problems. It focuses on a variant where the cost function is the same as in the minimum tree cover problem~\cite{Arkin2006}, and the goal of graph covering is similar to the tree cover problem in \cite{Arkin93}. This problem is motivated by the applications of the tree cover problem~\cite{Arkin93}, with an additional constraint on the weights of trees that need to be selected for covering tasks. This constraint is inspired by modern technologies such as drones and electric vehicles, which can travel limited distance on one full charge. The bounded forest cover problem is relevant to applications of the tree cover problem where the coverage task must be completed using devices with limited operational duration before refueling or recharging.

There are two popular variants of graph covering problems. The first variant seeks covering of the edges, while the second requires a covering of the vertices \cite{Arkin93,Kim:1989, Arkin2006}. Vertex cover belongs to the former category, while minimum spanning tree and traveling salesman problem are of the latter type. Given a weighted graph and a positive $\lambda$, the bounded tree cover problem is to find a collection of trees, each with weight at most $\lambda$, such the union of the vertices in the trees is the vertex set itself. Khani and Salvatipour gave a 2.5 approximation algorithm for the bounded tree cover \cite{KhaniS14}. The special case of when the tree is a path was studied by Levin et al.~\cite{Arkin2006} who gave a 3-approximation algorithm. The bounded path cover problem is a vehicle routing problem where each vehicle travels at a distance of at most $\lambda$ and the set of vehicles serves all the nodes. A variant occurs when $\lambda$ is unbounded, and there is a restriction on the number of vehicles (at most $k$). Here, given a weighted graph, the goal is to find a collection of $k$ paths that cover all the vertices; here, Wu et al. \cite{wu20233} gave a 3/2 approximation algorithm under the assumption that the edge weights satisfy triangle inequality and each vertex is visited exactly once. Suppose we choose the covering subgraph to be a cycle. In that case, we obtain the bounded cycle cover problem, where the vertices of a given graph have to be covered by cycles. The objective is to minimize the number of cycles subject to a maximum length. A 32/7 approximation for the bounded cycle cover problem is due to Yu et al. \cite{yu2019new}. The literature on covering graphs by subgraphs is vast, and its numerous approximation algorithms are known. Results in the following papers are the closest to the forest cover problem \cite{Arkin93,KhaniS14,VHN,EvenGKRS04,GPS23}. The main difference between the problems listed in this paragraph and the bounded forest cover problem considered here is the coverage constraint. The coverage constraint in bounded forest cover is on edges; each edge needs to be covered by some vertex in the vertex cover as opposed to covering all the vertices.




\subsection{Contributions}

Both the forest cover and the bounded forest cover problems are NP-complete. We study approximation algorithms for them. For the forest cover problem, we
give a probabilistic algorithm with  $2+\epsilon$ approximation ratio in Theorem~\ref{thm:real-lb}. We give a deterministic algorithm with $2$ approximation ratio in Theorem~\ref{th:2factor}. For the bounded forest cover problem we give a 6 approximation algorithm. This is the first study on the forest cover problem, and the approximation ratio in this paper is best expected given the conditional hardness results for vertex cover \cite{dinur2005hardness, austrin2009inapproximability}.
The bounded forest cover problem was first studied in \cite{GPS23}, where an 8-approximation algorithm was given. The result in Theorem~\ref{th_BFC} improves the approximation ratio to $6$.

We use the probabilistic method \cite{alon2016probabilistic} to obtain the $2+\epsilon$-factor approximation algorithm for the forest cover problem. First, we show using the method of dual fitting,in Theorem~\ref{thm:binary}, that the restriction of the forest cover problem admits a 2-approximation. Now, given a graph with weights in the interval $[0,1]$ on the edges, we create a family of LP relaxations that are easy to solve using the result in Theorem~\ref{thm:binary}. We show in Theorem~\ref{thm:real-lb} that the average solution to the family satisfies the LP dual of the problem. Each dual solution in the LP family has a corresponding 2-approximate integral solution. We pick the primal integral solution with the smallest value. Such a solution satisfies the primal constraint (to the original problem) and is guaranteed to exist, and is a $2+\epsilon$-approximate solution (Theorem~\ref{thm:real}). The novel use of the probabilistic method is an essential contribution to this paper, and it might have independent applications.

We present a deterministic algorithm with a 2-approximation ratio for the forest cover problem using LP rounding which is placed in Section~\ref{sec:det} due to space limit. The algorithm rounds the variables in the solution of the relaxed LP formulation to obtain a forest cover. Each connected component in the subgraph induced by the non-zero variables in the LP optimal solution is treated separately. A minimum spanning tree (MST) is constructed in each component; pendent vertices of MST with low fractional values and edges incident on such vertices are discarded without violating the covering constraint.
\\

\section{Forest cover problem (FC)}\label{sec:fc}

\noindent\textbf{FC Problem:~}Consider an undirected weighted graph $G=(V, E, w)$ where $w:E \rightarrow [0,1]$. A forest is an acyclic subgraph of a graph $G$. We denote the edges in a forest $F$ as $E(F)$. A vertex cover $C \subseteq V$ is a set of vertices such that for every edge $e=(u, v) \in E$, at least one of $u, v$ is in $C$. A forest cover of a graph $G$ is a forest in $G$ such that the vertices in the forest form a vertex cover. The weighted index (WI) of a forest $F=\{T_1, T_2, \ldots T_k\}$ is calculated $wi(F) = \sum_{i=1}^k |T_i| - \sum_{e \in E(F)} (1-w_e)$, where $|T_i|$ is the number of vertices in tree $T_i$. 
The objective is to find a forest cover for a given graph with a minimum weighted index. The decision version asks whether a given graph has a forest cover with WI at most $d$ for some non-negative real number $d$.

An approximation preserving reduction is given below to prove the following theorem.

\begin{theorem}\label{thm:npc} 
The forest cover (FC) problem is NP-complete.
\end{theorem}

\begin{proof}
Let $(G', k')$ be an instance of the vertex cover problem where $G'=(V_{G'}, E_{G'})$ be an undirected graph, and the objective is to answer whether $G'$ has a vertex cover of size at most $k'$. Given ($G', k'$) create an instance of forest cover problem ($\hat{G}, \hat{k}$) where $\hat{k}=k'$, $\hat{G}= (V_{\hat{G}}, E_{\hat{G}}, w_{\hat{G}})$  such that $V_{\hat{G}}=V_{G'}$,  $E_{\hat{G}}=E_{G'}$ and $w_{\hat{G}}:E_{\hat{G}} \rightarrow \{1\}$.  Note that, $G'$ has a vertex cover of size at most $k'$ iff $\hat{G}$ has a forest cover with weighted index at most $\hat{k}$.

 If $J$=$\{v_1,  v_2,  \cdots,  v_j\}$ of size $j$ is a vertex cover in $G'$ where $j \le k'$. Consider a forest $F=(V_F, E_F)$ in $\hat{G}$ where $V_F=J$ and $E_F =\emptyset$. $F$ is a forest cover of $\hat{G}$. The weighted index of $F$,  $wi(F)= \sum_{e\in E_F} w(e)+ c_F$  is equal to $|J|$, that is $F$ is a forest cover of $\hat{G}$ with WI at most $\hat{k}=k'$.
 
Conversely,  let $F=(V_{F},  E_F)$ be a forest cover of $\hat{G}$ with  $wi(F)\leq \hat{k}$. Since  $F$ is a forest cover of $\hat{G}$,  $V_{F}$ is a vertex cover of $\hat{G}$.  therefore,  $V_{F}$ is also a vertex cover in $G'$. 
If $c_F$ is the number of connected components in forest $F$. Then, the sum of the weights of edges in $E_F$ is equal to $wi(F)-c_F$. As the  weight of every edge in $\hat{G}$ is 1;  $|E_F|=wi(F)-c_F$. $F$ is a forest,  therefore,  $|E_F|=|V_F|-c_F$. So, $|V_F|=wi(F)$. Hence,  $G'$ has a vertex cover $V_F$ of size at most $k$. \qed
\end{proof}


\subsection{ILP formulation for Forest Cover \label{FC_LP}}
In this section,  we give an integer linear programming formulation for forest cover. This formulation is similar to the ones in \cite{Hajiaghayi:2006,Goemans:1993} for the Steiner tree problem. 

We use binary variables $x_i$ for each vertex $i \in V$ and $y_{ij}$ for each edge $(i, j) \in E$. The variable $x_i$ is set to $1$ if vertex $i$ is present in the forest cover. Otherwise, $x_i$ is $0$. Similarly, the variable $y_{ij}$ is set to $1$ if edge $(i, j)$ is in the forest cover, otherwise, $y_{ij}=0$. For $S \subseteq V$, we use $E(S)$ to refer to the edges with both endpoints in $S$.
\begin{align*}
\min \sum_{i\in V}x_i - \sum_{(i, j)\in E}y_{ij}(1-w_{ij}) \\
	  x_i + x_j & {\geq 1} & \quad{\forall (i, j) \in E} \\
	   x_i & \geq  y_{ij} & \quad{  \forall i \in V,   \forall (i, j) \in E} \\
	 \sum_{i \in S} x_i - \sum_{(i, j) \in E(S) }y_{ij} & {\geq 1} & \quad \forall S \subseteq V, \;s.t.\; E(S) \neq \emptyset\\
	   x_i \in & \{0,  1\}& \quad \forall i \in V \\
	  y_{ij} \in &\{0,  1\}& \quad \forall (i, j)\in E
\end{align*}
The first constraint ensures that at least one end vertex of each edge must be present in the solution. The second constraint ensures that an edge is present in the solution only if both end vertices are present. The third constraint ensures that cycles are absent. The objective function is the weighted index of the forest determined by the values of the variables $x,y$. The number of constraints in the above ILP is exponential. The following lemma establishes that the optimal solution of the corresponding relaxed LP can be obtained in polynomial time.


\begin{lemma}\label{separation}
The relaxed linear programming problem of the above ILP can be solved optimally using the ellipsoid method.
\end{lemma} 
\begin{proof}
To prove the polynomial time solvability, it is enough to show the existence of a polynomial time separation oracle for the constraints of the third type (the third constraint of the ILP in Section~\ref{FC_LP}). The description of the separation oracle is given below.

Given a solution $x^*, y^*$ to the LP, a separation oracle returns $S^* \subseteq V$ with $E(S)\ne \Phi$, such that 

$$\sum_{i\in S^*} x_i^* - \sum_{(i,j) \in E(S^*)} y_{ij}^* < 1,$$ 

if such an $S^*$ with $E(S^*)\ne \Phi$ exists. If no such $S^*$ exists, then all the constraints of the third type are satisfied, and $x^*,y^*$ is the optimal LP solution to the Forest Cover LP. 

If $x^*, y^*$ is the current solution. For each edge $(s,t)\in E(V)$, we define the following linear program $P_{s,t}$.

\begin{align*}
\min & \sum_{i\in V} x_i^* x_i - \sum_{(i,j) \in E} y^*_{ij} y_{ij} \\
    & y_{ij} = \min\{x_i, x_j\} & \forall (i,j) \in E\\
    &  y_{st} = 1 \\
    &  x_{s} = 1 \\
    &  x_{t} = 1 \\
    & x_i, y_{ij} \in \{0,1\} & \forall i\in V, (i,j) \in E
\end{align*}

Note that the $\min$ constraint can be linearized as $y_{ij} \le x_i$, $y_{ij} \le x_j$ and $y_{ij} =1$ if both $x_i, x_j$ are 1. To minimize the objective function, the above-mentioned ILP would implicitely try to enforce that $y_{ij} =1$ if both $x_i, x_j$ are 1. Therefore, every $\min$ constraint $y_{ij} = \min\{x_i, x_j\}$ can be replaced by two constraints $y_{ij} \le x_i$ and $y_{ij} \le x_j$. The integer solution corresponds to a subset of vertices in $S$, and all the edges $E(S)$ that are contained in $S$. If the objective value is $<1$ for at least one $P_{st}$, for $(s,t)\in E(V)$ then the corresponding set $S=\{i|x_i=1\}$ violates the constraint of the third type.

Each constraint has at most two variables for each of the above $P_{st}$. Hence, following the result in \cite{HochbaumN94}, $P_{s,t}$ can be solved in polynomial time.  \qed \end{proof}

\section{Probabilistic Algorithm for Forest Cover} 
In this section, we propose a probabilistic algorithm for forest cover with approximation factor arbitrarily close to $2$. The algorithm is described in two steps. In the first step, we present a deterministic $2$-approximation algorithm for forest cover where the weights on the edges are either 0 or 1.  In the second step, we use the algorithm for binary weights as a subroutine and give a probabilistic $(2+\epsilon)$- factor approximation algorithm for forest cover, where $\epsilon$ is a positive real close to 0.

\subsection{Binary Weights}

Let $G=(V,E)$ be a graph with binary weights on the edges, i.e.,  the edge weights are either 0 or 1. We say edge $e \in E(S)$, for $S \subseteq V$ if both the endpoints of $e$ are in the vertex set $S$.  We also use $u \in e$ to refer to the fact that edge $e$ is incident on a vertex $u$. Let us recall the primal integer program for the forest cover but this time we use different labels for the indices. We call this linear program as $P$.
\vspace{-3.5mm}
\begin{align}
\min  \sum_{u \in V} x_u   & -  \sum_{e \in E } y_e (1- w_e) \\ \label{cons:c1} 
x_u + x_v &\ge 1 &\quad \forall e=(u,v) \in E \\ \label{cons:c2} 
x_u  - y_e & \ge 0  \\
\label{cons:c3}
x_v  - y_e & \ge 0 &\quad \forall e=(u,v) \in E \\
 \sum_{u \in S} x_u - \sum_{e \in E(S)} y_e & \ge  1 &\quad \forall S \subseteq V , \;s.t.\; E(S) \neq \emptyset  \label{cons:c4} \\ 
 x_u, y_e &\in \{0,1\} &\quad \forall u \in V, e \in E
\end{align}
 The dual variables associated with the first set of constraints is $z_e$, $z_{ue}, z_{ve}$ are dual variables associated with the next two constraints, and the dual variable associated with the last set of constraints is $z_S$. 

The linear programming dual of the integer program above is 
\begin{align}
\max  \sum_{e \in E} z_e  + \sum_{S \subseteq V} z_S & \\
\label{c1}
\sum_{e : u \in e} z_e +  \sum_{e : u \in e} z_{ue}  + \sum_{S : u \in S} z_S & \le 1 & \forall u \in V\\  \label{c2}
\sum_{S : e \in E(S)} z_S + z_{ue} + z_{ve} &\ge (1- w_e) & \forall e \in E \\ 
z_e, z_{ue}, z_{ve}, z_S &\ge 0 &\forall e \in E,\; \forall u \in V,\;  S \subseteq V
\end{align}
Let $E_i$ be the set of edges of weight $i \in \{0,1\}.$ Let $V_0$ be the set of vertices incident on some edge in $E_0$ and $V_1 = V \setminus V_0$.  Let $E_i(V')$ be the set of edges in $E_i$ where $i\in \{0,1\}$ with both the endpoints in $V' \subseteq V$. Notice that there are edges of weight 1 with both the endpoints in $V_0$, and any such edge is not in $E_0(V_0)$.  The only edges with both endpoints in $V_1$ are of weight 1.
The edges with one endpoint in $V_0$ and the other in $V_1$ are all of weight 1.

$G_0(V_0, E_0(V_0))$ is the subgraph with vertices in $V_0$ and all the edges in $G_0$ are of weight $0$. Similarly, we define $G_1(V_1, E_1(V_1))$.  Stated otherwise, $G_0(V_0, E_0(V_0))$ is the subgraph of $G$ with edges with weight $0$, and all the vertices in $V_0$ are incident on some edges in $E_0$. The subgraph $G_1(V_1, E_1(V_1))$ contains vertices that are not incident on any weight 0 edges, and all edges in this subgraph have weight 1. 

Let the number of connected components $C_1, C_2, \ldots, C_k$ in $G_0(V_0, E_0(V_0))$ be $k$. For each connected component $C_i$ we identify a tree $T_i$ with $|C_i|-1$ edges. In the subgraph $G_1(V_1, E_1(V_1))$ we find a maximum cardinality matching $M$. 

We will show that $|M| +k$ is a lower bound on the value of the optimal solution to the forest cover problem. We will construct a feasible solution to the dual with-value $k + |M|$.

\begin{lemma} \label{lem:lb}
$|M| + k$ is a lower bound on the value of the optimal solution to the forest cover problem. 
\end{lemma}

\begin{proof}

Recall, the $k$ connected components of  $G_0(V_0, E_0(V_0))$ are $C_1, C_2, \ldots, C_k$. For each connected component $C_i$, we have a tree $T_i$ with $|C_i|-1$ edges. Starting with an initial solution in which all the dual variables are 0, compute a feasible solution as follows: 

\begin{itemize}
    \item For each set of vertices $S$ equal to some $C_i$, set $z_S = 1$.
    \item For each edge $e=(u,v)\in M$, set $z_S = 1$, where $S = \{u,v\}.$
\end{itemize}

Note that the sets $S \subseteq V$ for which $z_S=1$ are pairwise disjoint. The objective function value is $k+|M|.$ What remains to be shown is that the solution is feasible.

First we show that constraints given by \eqref{c1} are satisfied.
Each vertex $u\in V_0$ is in one connected component. Each vertex $u \in V_1$ is incident on at most one edge $e \in M.$ In both the cases 
$$\sum_{S : u \in S} z_S \le 1 \text{ and } \sum_{e : u \in e} z_e +  \sum_{e : u \in e} z_{ue} =0$$
Therefore, the first constraint is satisfied for all the vertices. Constraint \eqref{c2} is interesting only for edges with 0 as it is trivially satisfied for edges with weight 1. Each edge weight 0 is in some connected component $C_i$ (and only one). Therefore, $\sum_{S : e \in E(S)} z_S = 1$ for any edge with weight 0.

Since the solution is a feasible one, the LP relaxation of the primal has a value at least $k +|M|$ (by weak duality). The optimal value for the LP relaxation of the primal is a lower bound on the optimal value of forest cover. \qed
\end{proof}
This lower bound gives us a simple 2-approximation algorithm (Algorithm \ref{alg:FCB}) for the forest cover problem. There are two stages. In the first stage, we compute the connected components in $G_0(V_0, E_0(V_0))$ and create an assignment to the primal variables based on the connected components. In the second stage, we compute a maximum matching $G_1(V_1, E_1(V_1))$ and determine the values of the primal variables. The solution that we construct will be feasible. 
\begin{algorithm} [b]
    \caption{\textsc{ForestCoverBinary}($G$)}
    \label{alg:FCB}
    {Let the $k$ connected components of  $G_0(V_0, E_0(V_0))$ be  $C_1, C_2, \ldots, C_k$. For each connected component $C_i$ we identify a tree $T_i$ with $|C_i|-1$ edges.}\\
    \For{each vertex $u \in \cup_{i=1}^k C_i$}{
     {set $x_u =1$.}
    }
    \For{ each edge $e \in \cup_{i=1}^k T_i$}{

    {Set $y_e = 1$. The edges that are in this subgraph but not in any tree are assigned a value of $0$.}
    }
    {Find a maximum matching $M$ in $G_1(V_1, E_1(V_1))$ .}\\
    \For{each edge $e=(u,v) \in M$}{
    {Set  $x_u = 1$, $x_v = 1$, and $y_e = 1$. }
    }
    {For all the other vertices, set $x_u =0$.}
\end{algorithm}

The following lemma shows that the solution obtained according to Algorithm \ref{alg:FCB} is a feasible solution to the primal integer linear program.

\begin{lemma}
The solution $(x,y)$  returned by Algorithm \ref{alg:FCB} is a feasible solution of the primal integer linear program for forest cover. 
\end{lemma}
\begin{proof}
Every edge is incident on some vertex in $C_i$ or incident on an edge in $M$, therefore, constraint \eqref{cons:c1} is satisfied. Take any edge $e=(u,v)$ with $y_e =1$, both the variables $x_u. x_v$ are set to 1. So, \eqref{cons:c2} and \eqref{cons:c3} is satisfied. Finally, for any $S\subseteq V$, if $e\in S$ and  $y_e=1$ then both the endpoints of edge have $x_u,x_v$ set to 1: either the $y_e$ was set to 1 in the matching or in the construction of the connected components. Edges in $S$ with $y_e=1$ form a forest, therefore, constraint \eqref{cons:c4} is satisfied. \qed
\end{proof}

The solution constructed above is a feasible integer solution to the primal, and the objective function value is $k + 2|M|$. Therefore, we have the following theorem.

\begin{theorem} \label{thm:binary}
Algorithm \ref{alg:FCB} is a 2-factor approximation algorithm for the forest cover problem on graphs with binary weights.
\end{theorem}


\subsection{Real Weights}

Now we consider the case when the weights on the edges are in the closed interval $[0,1]$.
Let $\epsilon$ be a very small positive real close to 0 and \textcolor{black}{ $\delta=\epsilon^2$}.
Let us introduce a very small error $\delta$ to the objective function of the linear program $P$ as follows.
\begin{equation}\label{eq:error}
\sum_{u \in V} x_u   -  \sum_{e \in E } y_e (1- w_e-\delta)    
\end{equation}

where $\delta$ is a small real positive arbitrarily close to 0. 
With equation \ref{eq:error} as the objective function and the same set of constraints as in $P$, we call this linear program $P'$. For arbitrarily small $\delta$, $P$ and $P'$ admits the same optimal solution $(X^*,Y^*)$. Let $OPT$ and $OPT'$ be the optimal values of the objective functions of $P$ and $P'$, respectively. Then $OPT'=OPT+\delta\sum_e y^*_e$.

For $\epsilon$ very small, $\epsilon\sum_e y^*_e$ is smaller than $OPT$, therefore $\delta\sum_e y^*_e \le \epsilon OPT$. This implies that
$OPT'\le (1+\epsilon) OPT$.

Let $D'$ be the dual of $P'$. Then $D'$ has the same objective function and the constraint \ref{c1} as D. The constraint \ref{c2} of D is changed to the following inequality. 
\begin{equation}\label{eq:new}
\sum_{S : e \in E(S)} z_S + z_{ue} + z_{ve} \ge 1-w_e-\delta  ~~~\forall e \in E \\ 
\end{equation}
With the definition of $D'$, we are now ready to explain the algorithm in this section.

For any edge $e \in E$, let $W_e$ be an indicator variable which is 1 with probability $(1-w_e)$ (0 with probability $w_e$). We replace the RHS in the last constraint of the $D$  with this indicator variable. This gives a family of linear programs in which each edge has a weight of 0 or 1.
\begin{align}
\max  \sum_{e \in E} z_e  + \sum_{S \subseteq V} z_S & \\
\label{cc1}
\sum_{e : u \in e} z_e +  \sum_{e : u \in e} z_{ue}  + \sum_{S : u \in S} z_S & \le 1 & \forall u \in V\\  \label{cc2}
\sum_{S : e \in E(S)} z_S + z_{ue} + z_{ve} &\ge W_e& \forall e \in E \\ 
z_e, z_{ue}, z_{ve}, z_S &\ge 0 &\forall e \in E, S \subseteq V
\end{align}
Similarly, the last constrain of $D'$ becomes
\begin{equation}\label{eq:new1}
\sum_{S : e \in E(S)} (z_S + z_{ue} + z_{ve})  \ge W_e-\delta  ~~~\forall e \in E 
\end{equation}
For an experiment, we randomly generate the values $W_e \in \{0,1\}$ for all edges. This gives us an instance with binary edge weights in $\{0,1\}$. We can compute a lower bound for this instance and also an upper bound using the results in the previous section.

Suppose we run $m$ such experiments $E_1, \ldots,  E_m$. Each of these experiments gives us a feasible solution $z^i$ to the dual LP where $i\in \{1, \ldots, m\}$. Let $\overline{z}$ be the average value of the variables in all the solutions. 
\begin{theorem} \label{thm:real-lb}
    The average solution $\overline{z}$ over $m$ experiments, where $n$ is the number of edges, $m=n/(2\delta^2)$ and $0<\delta<1$, is a feasible solution to $D'$ with a high probability, and the objective function value given by the average solution of all the experiments, $\sum_{e \in E} \overline{z}_e  + \sum_{S \subseteq V} \overline{z}_S$
is a lower bound on $OPT'$, the minimum value of the objective function of $P'$.
\end{theorem}
\begin{proof}
Since $z^i$ is a feasible solution of $D$ corresponding to the $i$-th experiment, we have the following sequence of inequalities.
\begin{align*}
\sum_{e : u \in e} z^i_e +  \sum_{e : u \in e} z^i_{ue}  + \sum_{S : u \in S} z^i_S & \le 1 \\
\sum_{i=1}^m \sum_{e : u \in e} z^i_e +  \sum_{i=1}^m \sum_{e : u \in e} z^i_{ue}  + \sum_{i=1}^m \sum_{S : u \in S} z^i_S & \le m\\
 \sum_{e : u \in e} \sum_{i=1}^m z^i_e +  \sum_{e : u \in e} \sum_{i=1}^m  z^i_{ue}  + \sum_{S : u \in S}  \sum_{i=1}^m z^i_S & \le m\\
\sum_{e : u \in e} \overline{z}_e +  \sum_{e : u \in e} \overline{z}_{ue}  + \sum_{S : u \in S}  \overline{z}_S & \le 1
\end{align*}
This shows that  $\overline{z}$  satisfies the first constraint \eqref{cc1} of $D'$, as the first constrain is same for $D$ and $D'$. 

 For each $i$, the last constraint in $D$ is satisfied, summing it over all $i$ and after taking the average we get the following.
\begin{align*}
\sum_{S : e \in E(S)} z^i_S + z^i_{ue} + z^i_{ve} &\ge W^i_e \\
\sum_{S : e \in E(S)} \sum_{i=1}^m z^i_S + \sum_{i=1}^m z^i_{ue} + \sum_{i=1}^m z^i_{ve} &\ge \sum_{i=1}^m W^i_e \\
\sum_{S : e \in E(S)} \overline{z}_S + \overline{z}_{ue} + \overline{z}_{ve} &\ge \frac{1}{m}\sum_{i=1}^m W^i_e
\end{align*}
From Chernoff-Hoeffding bound\cite{Doerr:2020}, we know that 
$$ 
Pr[\frac{1}{m}\sum_{i=1}^m W^i_e \le (1-w_e) - \delta ] \le \frac{1}{e^{2m \delta^2}}
$$

If we choose $m=n/(2\delta^2)$ where $n$ is the number of edges then $Pr[\frac{1}{m}\sum_{i=1}^m W^i_e > (1-w_e) - \delta ] \le 1-\frac{1}{e^{n}}$. Each edge is set to 0 or 1 with probability $w_e$ independent of other edges. Since there are $n$ constraints, the probability that the above inequality satisfies for every edge is close to $1$.
Hence, $\overline{z}$ is a feasible solution of $D'$ with very high probability. This implies that the objective function value of $D'$ for the solution $\overline{z}$ is a lower bound of the optimal value of $P'$ with high probability. \qed
\end{proof}

\noindent

Below, we describe the algorithm for real weights.

{\color{black}{{\bf Algorithm:} For each experiment $E_i$ we have a feasible solution $D_i$ to the dual and a solution $P_i$  to the primal problem (for integer weights) given the algorithm (Algorithm~\ref{alg:FCB}) in the previous section where $W^i_e$ be the indicator variable with value in $\{0,1\}$ which takes the value $1$ with probability $1-w_e$ in the $i^{th}$ experiment. The algorithm returns the solution $P_j$, which has the minimum objective function value among all $P_i$'s.}}

The following theorem proves that the above algorithm is an approximation algorithm with factor close to 2 with very high probability.
\begin{theorem} \label{thm:real}
The proposed algorithm is a $(2+\epsilon)$ factor approximation algorithm for Forest cover problem with high probability.
\end{theorem}
\begin{proof}

Let $\overline{P} (\overline{D})$ be the average value of the primal (dual) solutions. Then, $\overline{P} = \sum_{i=1}^m P_i/m$, 
where $$P_i = \sum_{u \in V} x^i_u    -  \sum_{e \in E } y^i_e W^i_e +  2|M_i|$$
and the value of the dual solution $D_i$ is
$$
D_i = \sum_{e \in E}z^i_e  + \sum_{S \subseteq V} z^i_S
$$
and $M_i$ is the maximum matching in $G_1(V_1, E_1(V_1))$; in the graph obtained in the $i^{th}$ experiment.  The average value $\overline{D} = \sum_{i=1}^m D_i/m.$

Since each primal solution is a two approximate solution (to the instance) $P_i \le 2D_i$, the  following inequality holds:

$$m \overline{P} = \sum_{i=1}^m \sum_{u \in V} x^i_u    - \sum_{i=1}^m \sum_{e \in E }  y^i_e  W^i_e + \sum_{i=1}^m 2|M_i| \le 2m \overline{D}$$ 

Let us rewrite the middle term. The binary value of $y^i_e$ is given by the Algorithm in the previous section and it is fixed. The expected value of  $y^i_e W^i_e$,
$$\mathbb{E}[y^i_e W^i_e]= y^i_e \mathbb{E}[W^i_e] = y^i_e (1-w_e).$$  Therefore, we can rewrite the previous equation (in expectation) as 
\begin{align*}
m \overline{P} &= \sum_{i=1}^m \sum_{u \in V} x^i_u    - \sum_{i=1}^m \sum_{e \in E }   y^i_e (1-w_e) + \sum_{i=1}^m 2|M_i| \le 2m \overline{D}   \\
\overline{P} &= \sum_{u \in V} \overline{x}_u    -    \sum_{e \in E }   \overline{y}_e (1-w_e) +  2|\overline{M}| \le 2 \overline{D} 
\end{align*}


Since $D$ and $D'$ have the same objective function, the value of the objective functions of both $D$ and $D'$ are same for the average solution $\overline{z}$. Therefore $\overline{P} \le 2\overline{D'}$. Also, $\overline{D'}$ is a lower bound of $Opt'$ with high probability. Hence, with high probability, $\overline{D'} \le Opt' \le Opt(1+\epsilon)$. Hence, with high probability, $\overline{P} \le 2\cdot Opt\cdot (1+\epsilon)$. 

Each $x^i_u, y^i_u$ for $i\in \{1, \ldots, m\}$ is a feasible solution, 
so there is a feasible primal solution with value at most twice the average value of the dual solution. Since we have selected the solution which gives the minimum objective function value among all $x^i_u, y^i_u$ for $i\in \{1, \ldots, m\}$, our proposed solution is  an integer solution  to the primal with value at most the average value of the primal solutions. Therefore, we have a $(2+\epsilon)$-approximate solution with probability close to 1. \qed
\end{proof}

\section{Deterministic Algorithm for Forest Cover using LP rounding \label{sec:det}}
In this section, we present a deterministic algorithm with a 2-approximation ratio for the forest cover problem using LP rounding.\\ 

\noindent\textbf{Algorithm and Analysis Idea:} The algorithm rounds the variables in the solution of the relaxed LP formulation to obtain a forest cover. Each connected component in the subgraph induced by the non-zero variables in the LP optimal solution is treated separately. A minimum spanning tree (MST) is constructed in each component; pendent vertices of MST with low fractional values and edges incident on such vertices are discarded without violating the covering constraint.
\\
Theorem~\ref{th:2factor}  gives a solution constructed using a deterministic algorithm to the forest cover problem that is 2 approximate. The proof of this theorem is based on the bounds established in Lemma~\ref{Th_main} on the in-part rounded solution in each subgraph component. This theorem shows that the cost of the rounded solution for connected component is at most two times the cost of the optimal fractional solution for that component. The proof of Lemma~\ref{Th_main} relies on several lemmas and corollaries.  During the analysis, we partition the variables based on the types of edges (tree edges, non-tree edges, deleted edges, etc.) in a component. This approach helps us to better understand the behaviour of the algorithm and simplify the proof. We ensure that the number of components doesn't increase during rounding. We first establish an upper bound on the cost of a tree in the rounded solution based on the optimal fractional cost of the component in which that tree was constructed. We prove the upper bound by establishing a loop invariant property that holds at the end of every iteration during the construction of the MST. 
\\
A detailed analysis considering all types of edges in the rounding solution is required because we do not perform rounding of variables for edges based on the fraction values. Rounding is done based on whether edges are chosen in MST construction, which takes edge weights into account. There may be a large integrality gap for some edge variables due to this reason. Therefore, the bound on the overall cost of the forest cover is established using both variables for edges and vertices. We establish loop invariant properties that hold at the end of every iteration while constructing minimum spanning trees. The analysis proceeds for all edges of a component in the order given by Kruskal's algorithm~\cite{Kruskal}. The sequence of edges of a component is partitioned into two carefully chosen groups; if a suffix (a sub-sequence of edges) exists which may violate the bound due to a larger integrality gap. We consider the corresponding terms for the edges and vertices of both partitions separately to derive inequalities. We then add the derived inequalities to prove Lemma~\ref{Th_main}. Deleting non-essential vertices and edges reduces the cost of forest cover without increasing the number of components. Therefore, this step may only improve the approximate solution.\\


Next we give the deterministic 2-approximation algorithm for forest cover.

Let ($x^*, y^*$) be the optimal solution to the LP formulation of forest cover given in Section~\ref{FC_LP}. Consider a subgraph $G^*=(V_{G^*}, E_{G^*})$ of $G=(V, E)$ formed by the vertex set $V_{G^*} = \{i \in V~|~ x_i^* > 0\}$ and edge set $E_{G^*}$=\{$(i, j)\in E~|~ y_{ij}^*>0$\}. A connected component is a non-empty subgraph in which every pair of vertices is connected by a path. Let $\mathbb{C} = \{C_1,  C_2, \cdots,  C_k\}$ be the set of all connected components in the graph $G^*$. Let $\hat{C}$ be the set of all isolated vertices in $G^*$, i.e., for each vertex $i \in \hat{C}$, $x_i^* > 0$ and $\sum\limits_{j:(i, j)\in E}y_{ij}^* = 0$.

For each connected component $C \in \mathbb{C}$, let $C=(V_{C}, E_{C})$. We then partition $V_{C}$ into two subsets:  
$V_{C}^{\ge 0.5}$=\{$i \in V_{C} ~|~ x_i^*\ge 0.5$\} and $V_{C}^{< 0.5}$=\{$i \in V_{C} ~|~ x_i^*<0.5$\}.  Similarly, we partition $\hat{C}$ into two subsets: $\hat{C}^A$=\{$i \in \hat{C}~|~ x_i^*\ge 0.5$\} and $\hat{C}^B$=\{$i \in \hat{C}~|~ x_i^*<0.5$\}.

\begin{algorithm}
{Solve the LP relaxation. Let $(x^*,  y^*)$ be the optimal solution.}

Consider the subgraph $G^*=(V_{G^*}, E_{G^*})$ as defined earlier.

Find $\mathbb{C} = \{C_1,  C_2, \cdots,  C_k\}$,  the set of all non-empty connected components in the graph $G^*$.

Find $\hat{C}$, the set of isolated vertices in the graph $G^*$. Partition $\hat{C}$ into $\hat{C}^A$and $\hat{C}^B$  as defined above.

Let $F=(V_F, E_F)$ be an empty forest. Set $F= \hat{C}^A$.
        
\For { each $C=(V_{C},  E_{C}) \in \mathbb{C}$}
      { Partition $V_{C}$ into $V_{\ge 0.5}$=\{$i \in V_{C}| x_i^*\ge 0.5$\} and $V_{C}^{< 0.5}$=\{$i \in V_{C}| x_i^*<0.5$\}.\\
      
      Find an MST $M$ on $C$ using Kruskal's algorithm.\\
      
      In tree $M$, delete all pendant vertices from the set $V_{C}^{< 0.5}$. 
      
      Let $T$ be the remaining tree after deletion.\\
      
      $F \leftarrow F \cup T.$}

\For{each $i \in V$}{
    \uIf{$i \in V_F$}
        { $x_i' \leftarrow 1$}
     \Else{$x_i' \leftarrow 0$}
    }
    
\For{each $(i, j) \in E$}{
    \uIf{$(i, j)\in E_F$}
        { $y_{ij}' \leftarrow 1$}
     \Else { $y_{ij}' \leftarrow 0$}
    }

$(x', y')$ is the rounded solution. 

Return $F$ and $(x', y')$.

\caption{LP rounding for Forest Cover}
\label{algo1}
\end{algorithm} 

Algorithm~\ref{algo1} computes a forest cover using the optimal solution derived from the LP relaxation. First, we create a subgraph $G^*$ by including only the vertices and edges with non-zero values in the LP solution. Next, we consider every non-empty connected component $C$ in $G^*$ with at least one edge. Using Kruskal's method, we create a minimum spanning tree (MST) $M$ for each component $C$. The vertices in $M$ can be divided into two groups, $V_{C}^{\ge 0.5}$ and $V_{C}^{< 0.5}$, based on their respective values in the optimal solution. The tree $M$ constitutes a tree cover for $C$. Finally, we remove all the pendant (degree one) vertices in $V_{C}^{< 0.5}$ from $M$, along with any edges that are incident on such vertices.

We define $V_{D}$ as the set of all vertices from $V_{C}^{< 0.5}$ that have a degree of one in $M$. Similarly, we define $E_{D}$ as the set of all edges removed from $M$ due to the deletion of vertices from the set $V_{D}$. By deleting all such vertices and edges, we obtain a tree $T=(V_{T}, E_{T})$ from $M$. Since deleting pendant vertices and corresponding edges from $V_{C}^{<0.5}$ does not violate coverage of $C$, $T$ is also a tree cover of $C$. We repeat this process for each non-empty component $C_\ell$ in $\mathbb{C} = \{C_1, C_2, ..., C_k\}$, and obtain a tree cover $T_\ell$ for each component. Finally, we return $F= \bigcup_{\ell=1}^k T_\ell \bigcup \hat{C}^A$ as the forest cover of the graph $G$. An illustration of Algorithm~\ref{algo1} is given below.

\begin{figure}[h]
         
         \subfloat[Graph $G^*$\label{fig:subim1}]{%
        \includegraphics[width=0.3\textwidth]{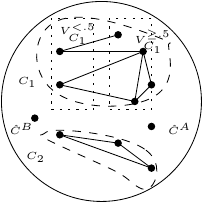}%
         }\hfill
        \subfloat[$MST$ in each non-empty component of $G^*$\label{fig:subim2}]{%
        \includegraphics[width=0.3\textwidth]{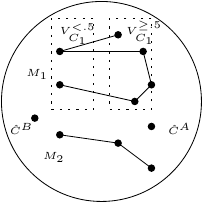}%
        }\hfill
      \subfloat[Forest $F$ in $G^*$\label{fig:subim3}]{%
      \includegraphics[width=0.3\textwidth]{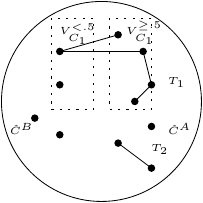}%
       }
  \caption{Steps to find a forest cover $F$ in $G^*$}
\label{fig:image2}

\end{figure}  
Figure~\ref{fig:image2} illustrates Algorithm~\ref{algo1}. An instance of a graph $G^*$ with two non-empty connected components $C_1$ and $C_2$ and an isolated vertex in each set $\hat{C^A}$ and $\hat{C^B}$ is shown in Figure~\ref{fig:subim1}. The minimum spanning trees $M_1$ and $M_2$ in components $C_1$ and $C_2$, respectively, are in Figure~\ref{fig:subim2}. The trees $T_1$ and $T_2$ constructed by deleting the pendant vertices in $V_{D_1}$ and $V_{D_2}$  are shown in Figure~ \ref{fig:subim3}.


Forest $F$ is constructed by taking union of all trees $T$, for $1\leq \ell \leq k$,  and all vertices in $\hat{C}^A$. The algorithm computes two vectors $x'$, of size $|V|$, and $y'$ of size $|E|$ as follows. For each $i \in V$,  $x'_i=1$ if $i \in F$, else $x'_i=0$. For each $(i, j) \in E$,  $y'_{ij}=1$ if $(i, j) \in F$, else $y'_{ij}=0$. The algorithm outputs ($x'$, $y'$) as the rounded solution which is a forest cover.

\subsection{Correctness and Approximation Ratio}

The following lemma shows that Algorithm \ref{algo1} returns a forest cover.

\begin{lemma}
Set $F$ is a forest cover of $G$.
\end{lemma}
\begin{proof}
Given the covering constraints in the LP formulation, for each $(i, j) \in E$, there must be at least one vertex $i$ where $x_i^* \geq 0.5$. Thus, if we consider the set $V^{\geq 0.5}$ of all vertices of $V$ with $x_i^* \geq 0.5$, we have a vertex cover of $G$. The vertex set $V_{F}$ of $F$ contains all vertices of $V$ with $x_i^* \geq 0.5$ and some vertices with $x_i^* < 0.5$. Also there are no cycles in $F$. Therefore, we have a forest cover of $G$. \qed
\end{proof}

Algorithm~\ref{algo1} generates a tree $T$ for each non-empty component of $G^*$. These trees are then combined to form a forest. The theorem below states that the cost  (weighted index $wi$)  of forest cover computed the algorithm for each non-empty component $C\in\mathbb{C}$ of $G^*$ is at most twice the cost of the fractional solution restricted to the component $C$.

\begin{lemma}\label{Th_main}
Given a graph $G^*$ based on the optimal solution of the  LP relaxation for the forest cover problem on a weighted graph $G=(V,  E,  w)$,  the following inequality holds for  each non-empty connected component $C$ in $G^*$

$$\sum_{i \in V_{C}}{x_i^*} - \sum_{(i, j) \in E_{C}}{y_{ij}^*}(1 - w_{ij}) \geq \frac{1}{2}\bigl[\sum_{i \in V_{T}}{x'_i} - \sum_{(i, j) \in E_{T}}{y'_{ij}}(1 - w_{ij}) \bigr].$$

\end{lemma}

\textcolor{black}{ The proof of Lemma~\ref{Th_main} is in  appendix \ref{app:main}. The proof is routine but long. We use it to show that Algorithm \ref{algo1} is a $2$-factor approximation algorithm (Theorem \ref{th:2factor}).
 }

Recall that the weighted index ($wi$) of a forest $F$ is calculated as follows: $wi(F)= \sum_{e\in E_F} w(e)+ c_F$, where $c_F$ stands for the number of connected components in the forest $F$, and $E_F$ is the set of edges in $F$. Let $OPT_{FC}$ be the optimal solution of the forest cover problem over graph $G$, and let $wi(OPT_{MC})$ be its weighted index. The approximate solution $APX_{FC}$ obtained using the LP rounding, Algorithm~\ref{algo1}, has a weighted index value of $wi(APX_{FC})$.

\begin{theorem}\label{th:2factor} $wi(OPT_{FC}) \geq \frac{1}{2}wi(APX_{FC})$.
\end{theorem} \label{mwcfc-approx}
\begin{proof}
In the algorithm we round up each isolated vertex $i$ in $\hat{C}^A$ such that $x_i^* \geq \frac{1}{2}x'$, and round down each isolated vertex $i$ in $\hat{C}^B$ such that $x_i^* \geq x' = 0$, then $\sum\limits_{i \in \hat{C}^A}{x_i^*} \geq \frac{1}{2}\sum\limits_{i \in \hat{C}^A}{x_i'}$. Similarly, $\sum\limits_{i \in \hat{C}^B}{x_i^*} \geq 0$.

By summing over all non-empty connected components in graph $G^*$ and the set of isolated vertices of $\hat{C}^A$ and $\hat{C}^B$, we obtain the following equation.
\begin{multline*}
    \sum\limits_{C \in \mathbb{C}}{\sum\limits_{i \in V_{C}}{x_i^*}}-\sum\limits_{C \in \mathbb{C}}{\sum\limits_{(i, j)\in E_{C}}{y_{ij}^*}(1-w_{ij})} + \sum\limits_{i \in \hat{C}^A}{x_i^*} + \sum\limits_{i \in \hat{C}^B}{x_i^*} \\
    \geq \frac{1}{2}[ \sum\limits_{T \in F}{\sum\limits_{i \in V_{T}}{x_i'}}-\sum\limits_{T \in F}{\sum\limits_{(i, j) \in E_{T}}{y_{ij}'}(1-w_{ij})}] + \frac{1}{2}\sum\limits_{i \in \hat{C}^A}{x_i'}
\end{multline*}

\begin{equation}\label{final_eq}
    \sum_{i \in V_{G^*}}{x_i^*} - \sum_{(i, j) \in E_{G^*}}{y_{ij}^*}(1-w_{ij}) \geq \frac{1}{2}\sum_{i \in V_{F}}{x_i'} -\frac{1}{2}\sum_{(i, j) \in E_{F}}{y_{ij}'}(1-w_{ij})
\end{equation}
Recall,  $x_i^* > 0$ for each $i \in V_{G^*}$ and $x_i^* = 0$ for each $i \in V_{G} \setminus V_{G^*}$. Similarly, $y_{ij}^* > 0$ for each $(i, j) \in E_{G^*}$ and $y_{ij}^* = 0 $ for each $(i, j) \in E_{G} \setminus E_{G^*}$. Hence,  we have 
$$\sum_{i \in V_{G}}{x_i^*} - \sum_{(i, j) \in E_{G}}{y_{ij}^*}(1-w_{ij}) = \sum_{i \in V_{G^*}}{x_i^*} - \sum_{(i, j) \in E_{G^*}}{y_{ij}^*}(1-w_{ij}).$$
The weighted index($wi$) of a forest $F$ is defined as $wi(F)= \sum_{e\in E_F} w(e)+ c_F$, where $c_F$ is the number of connected components in $F$ and $E_F$ is the set of edges in $F$. The solution of relaxed LP formulation of the forest cover problem provides a lower bound on $wi(OPT_{FC})$. Hence, 
$$wi(OPT_{FC}) \geq \sum\limits_{i \in V_{G}}{x_i^*} - \sum\limits_{(i, j) \in E_{G}}{y_{ij}^*}(1-w_{ij}).$$
If $APX_{FC}$ is the approximate solution of the forest cover problem with weights index $wi(APX_{FC})$ obtained using the  Algorithm~\ref{algo1}, then using \eqref{final_eq}, we get the desired bound. 
$wi(OPT_{FC}) \geq \sum\limits_{i \in V_{G}}{x_i^*} - \sum\limits_{(i, j) \in E_{G}}{y_{ij}^*}(1-w_{ij}) \geq \frac{1}{2}\sum_{i \in V_{F}}{x_i'} -\frac{1}{2}\sum_{(i, j) \in E_{F}}{y_{ij}'}(1-w_{ij}) = \frac{1}{2} wi(APX_{FC}).$\qed
\end{proof}

\subsection{Proof of Lemma~\ref{Th_main}} \label{app:main}

The steps involved in proving Theorem~\ref{Th_main} are as follows: we first derive Corollary~\ref{cor7}, which establishes a partial relationship between the left-hand side (LHS) and the right-hand side (RHS) of the inequality given in Lemma~\ref{Th_main}, without the terms multiplied by $w_{ij}$. This corollary is derived from Lemma~\ref{lemma6.3}, which proves that the relationship stated in the corollary is a loop invariant property that holds at the end of every iteration of Kruskal's algorithm~\cite{Kruskal} used in Algorithm~\ref{algo1} for constructing the MST $M$ on a component $C$.

Next we derive Lemma~\ref{lemma6.4} from Corollary~\ref{cor7},  Corollary~\ref{cor6.5} and Lemma~\ref{lemma6.3} by splitting the terms in the LHS of the inequality stated in Corollary~\ref{cor7} and Lemma~\ref{lemma6.3}, respectively, based on the types of edges (remaining tree edges, non-tree edges, deleted edges etc.) in the non-empty component.
The next step is to prove Lemma~\ref{lemma6.6}, which states that if the terms involving edges in LHS of the inequality stated in Lemma~\ref{lemma6.4} are multiplied by (1-$w_{ij}$), then the modified LHS remains greater than zero. 
Finally, we prove Lemma~\ref{Th_main} using Lemma~\ref{lemma6.6} by adding the terms given in the RHS of the inequality stated in Lemma~\ref{Th_main} to both sides of the inequality stated in Lemma~\ref{lemma6.6}. 

If $C=(V_{C}, E_{C})$ is a connected component in $G^*$ (with at least one edge), then Algorithm~\ref{algo1} computers MST, $M$ on $C$. After deleting all the pendant vertices of $M$ from the set $V_{C}^{< 0.5}$, we get the tree $T$. $V_{D}$ is the set of vertices that were deleted from $M$, while $E_{D}$ is the set of edges that were deleted along with the vertices of $V_{D}$. $E_{NT}$ is the set of edges $E_{C}\setminus E_{M}$.


 Let $SoE_{C}=(e_1,  e_2, \cdots, e_{|E_{C}|})$ be the ordered sequence of edges in $E_{C}$ that were added by Kruskal's algorithm  while forming MST $M$ on $C$. The edges in $SoE_{C}$ are arranged in non-decreasing order of weights, i.e., $w_{e_p}\leq w_{e_q}$ if $p\leq q$.  
We define $SoE_{C}^{(1, r)}$ as the sub-sequence of $SoE_{C}$ that contains the first $r$ consecutive edges of $SoE_{C}$. Let an edge set $E_{H_{r}}$ contains these $r$ edges.  
Furthermore, let $V_{H_{r}}$ be the set of all vertices that are endpoints of edges in $E_{H_{r}}$. The edge set $E_{H_{r}}$ over vertex set $V_{H_{r}}$ forms a sub-graph $H_{r}$ of $C$, where $H_{r}$ has a set of $k$ non-empty connected components $S_{r}^1, \cdots, S_{r}^k$.


Let $M_r^f = (V_{M_r^f}, E_{M_r^f})$, where $E_{M_r^f}$ is the set of edges present in $S_r^f$ that also appear in the minimum spanning tree $M$, and $V_{M_r^f}$ is the set of all vertices that are endpoints of edges in $E_{M_r^f}$.
Similarly, let $T_r^f = (V_{T_r^f}, E_{T_r^f})$, where $E_{T_r^f}$ is the set of edges present in $S_r^f$ that also appear in the tree $T$, and $V_{T_r^f}$ is the set of all vertices that are endpoints of edges in $E_{T_r^f}$.
Let $E_{NT_r^f}$ be the set of edges in $S_r^f$ that are not in $M_r^f$. Let $V_{D_r^f}$ be the set of vertices that are in $S_r^f$ as well as in $V_D$. Finally, let $E_{D_r^f}$ be the set of all edges deleted from $M$ due to the deletion of vertices in the set $V_{D_r^f}$.

Then the following lemma holds true, which states a partial relationship, without involving the terms multiplied by $w_{ij}$, between the LHS and the RHS  of the equation given in Lemma~\ref{Th_main}. This relationship holds true at the end of every iteration of Kruskal's algorithm while constructing a minimum spanning tree $M$ on $C$.
 

\begin{lemma} \label{lemma6.3}
In any non-empty connected component $C$ of $G^*$,  If $H_{r}$ is a sub-graph of $C$ formed by taking the first $r$ edges considered by Kruskal's algorithm while constructing minimum spanning tree $M$ on $C$,   and only the corresponding end vertices. Then,  the following relationship holds:

$$\sum\limits_{i\in V_{H_{r}}}{x_i^*} - \sum\limits_{(i, j)\in E_{H_{r}}}{y_{ij}^*} \geq \sum\limits_{f=1}^k\left({\sum\limits_{i\in V_{T_{r}^f}}}{x_i'}- {\sum\limits_{(i, j)\in E_{T_{r}^f}}}{y_{ij}'}\right).$$

\end{lemma}

\begin{proof} $H_{r}$ has $k$ non-empty connected components $S_{r}^{1},  \cdots,  S_{r}^{k}$,  therefore,  
$\sum\limits_{i\in V_{H_{r}}}{x_i^*} - \sum\limits_{(i, j)\in E_{H_{r}}}{y_{ij}^*}$ can be rewritten as $\sum\limits_{f=1}^k ({\sum\limits_{i \in V_{{S_{r}^f}}}}{x_i^*}- {\sum\limits_{(i, j) \in E_{S_{r}^f}}}{y_{ij}^*})$.
We prove it by contradiction. Let us assume that  $$\sum\limits_{f=1}^k ({\sum\limits_{i \in V_{{S_{r}^f}}}}{x_i^*}- {\sum\limits_{(i, j) \in E_{S_{r}^f}}}{y_{ij}^*}) < \sum\limits_{f=1}^k({\sum\limits_{i\in V_{T_{r}^f}}}{x_i'}- {\sum\limits_{(i, j)\in E_{T_{r}^f}}}{y_{ij}'})$$.

This implies that there exists at least a sub-component $S_{r}^h$ in $H_{r}$,  for $1 \leq h \leq k$,  in which $$\sum\limits_{i\in V_{S_{r}^h}}{x_i'}- {\sum\limits_{(i, j)\in E_{S_{r}^h}}{y_{ij}'} }< \sum\limits_{i\in V_{M_{r}^h}}{x_i'}- {\sum\limits_{(i, j)\in E_{M_{r}^h}}}{y_{ij}'}$$.


Recall that $V_{S_{r}^h}$ denotes the vertex set of the sub-component $S_{r}^h$, which is formed due to the edge set $E_{S_{r}^h}$.  $E_{S^h}$ consists of all the edges in $G^*$ that connect vertices within $V_{S_{r}^h}$. Given the constraint of the LP formulation, we have a set of vertices $V_{S_{r}^h} \subseteq V$ such that

\begin{equation}\label{eqn4}
\sum\limits_{ i \in V_{S_{r}^h}}x_i^* - \sum\limits_{(i, j) \in E_{S^h}}y_{ij}^* \geq 1 
\end{equation}


We only include the $r$ least weight edges of $C$ in the set $E_{H_{r}}$. However, there may be some high-weight edges in $E_{G^*}$ that have both end vertices in the set $V_{S_{r}^h}$, but are not present in the set $E_{S_{r}^h}$. The set $E_{S_{r}^h}$ consists of edges in the sub-graph of $G$ induced by the vertex set $V_{S_{r}^h}$, which means that $E_{S_{r}^h} \subseteq E_{S^h}$. We can rewrite \eqref{eqn4} as follows:

\begin{equation}\label{eqn5}
  \sum\limits_{ i \in V_{S_{r}^h}}x_i^* - \sum\limits_{(i, j) \in E_{S_{r}^h}}y_{ij}^* \geq 1 
\end{equation}


Next,  we prove that the sub-graph $M_{r}^h$ is a spanning tree on $S_{r}^h$. Suppose that the sub-graph $M_{r}^h$ is not a spanning tree on $S_{r}^h$. Note that $M_{r}^h$  is also a sub-graph of MST $M$ on $C$. Therefore,  edges in $M_{r}^h$ do not form a cycle. This means that for $M_{r}^h$ to not be a spanning tree on $S_{r}^h$, the graph ($V_{S_{r}^h}$, $E_{M_{r}^h}$) must have more than one connected component. It is important to note that these components are connected to each other in $S_{r}^h$ through edges that are not present in $M_{r}^h$. In $S_{r}^h$, each connected component in the graph ($V_{S_{r}^h}$, $E_{M_{r}^h}$) is connected to another such component via an edge that is present in $S_{r}^h$, but not in $M_{r}^h$.

\begin{figure}[h!]
\centering
\includegraphics[scale=1]{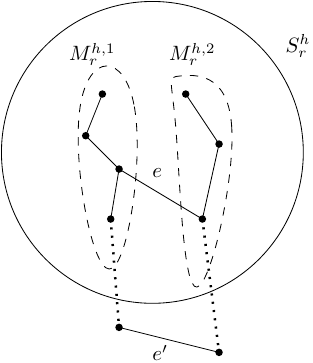}
\caption{The sub-graphs of $M_{r}^h$ in the sub-component  $S_{r}^h$.}
\label{fig:universe}
\end{figure}

Consider the graph ($V_{S_{r}^h}$,  $E_{M_{r}^h}$) and let $M_{r}^{h, 1}$ and $M_{r}^{h, 2}$ be two of its connected components that are connected in $S_{r}^h$ by some edges. Let $e=(i, j)\in E_{S_{r}^h}$ be the edge with the least weight, connecting $M_{r}^{h, 1}$ and $M_{r}^{h, 2}$, where $i \in V_{M_{r}^{h, 1}}$ and $j \in V_{M_{r}^{h, 2}}$. Although $i$ and $j$ are disconnected in ($V_{S_{r}^h}$,  $E_{M_{r}^h}$) as $e\notin E_{M_{r}^h}$, they are still connected in $M$. Therefore, a path $P(i, j)$ must exist from node $i$ to node $j$ in $M$.
 
The weight of every edge $e'$ on the path $P(i, j)$, denoted by $w(e')$, should be less than or equal to the weight of edge $e$.  Furthermore, every edge $e'$ on the path $P(i, j)$ with a weight equal to $w(e)$ must have been considered before edge $e$ in Kruskal's algorithm. Otherwise, $e'$ would not have been selected in $M$, and $e$ would have been chosen in $M$, which is not possible. Therefore, all edges $e'$ on the path $P(i, j)$ must have been considered before edge $e$ by Kruskal's algorithm for finding $M$. 

If $e$ is in $H_r$, then all edges $e'$ on the path $P(i, j)$ must also be present in $H_r$ and further in $S_r^h$ connecting node $i$ and node $j$. Since all the edges in the path $P(i, j)$ are in $M$, node $i$ and node $j$ will also be connected via path $P(i, j)$ in the graph ($V_{S_r^h}$, $E_{M_r^h}$). This contradicts the assumption that node $i$ and node $j$ belong to two different connected components $M_r^{h,1}$ and $M_r^{h,2}$, respectively, which are not connected in the graph ($V_{S_r^h}$, $E_{M_r^h}$). Hence, $M_r^h$ is a spanning tree on $S_r^h$.

After removing the vertices of the set $V_{D_{r}^h}$ and the edges of $E_{D_{r}^h}$ from $M_{r}^h$, the remaining sub-graph $T_{r}^h$ is a tree over the vertex set $V_{T_{r}^h}$. The number of edges in $T_{r}^h$ is exactly one less than the number of vertices, i.e. 
$$|V_{T_{r}^h}| - |E_{T_{r}^h}| = 1, \implies \sum\limits_{i\in V_{T_{r}^h}}{x_i'}- {\sum\limits_{(i, j)\in E_{T_{r}^h}}}{y_{ij}'} = 1$$. 

This is due to the rounding algorithm,  which picks all the vertices and edges of $T$ by rounding up the values corresponding to the vertices and edges of $T$ in the fractional solution.

$$ \sum\limits_{ i \in V_{S_{r}^h}}x_i^* - \sum\limits_{(i, j)\in E_{S_{r}^h}}y_{ij}^* \geq 1 = \sum\limits_{i\in V_{T_{r}^h}}{x_i'}- {\sum\limits_{(i, j)\in E_{T_{r}^h}}}y_{ij}'.$$

Therefore,  \eqref{eqn5} can be re-written as 
\begin{equation}\label{eqn6}
\sum\limits_{ i \in V_{S_{r}^h}} x_i^* - \sum\limits_{(i, j) \in  E_{S_{r}^h}}y_{ij}^* \geq  \sum\limits_{i\in V_{T_{r}^h}} x_i'- \sum\limits_{(i, j)\in E_{T_{r}^h}}y_{ij}' 
\end{equation}

This contradicts the assumption that there exists at least a sub-component $S_{r}^h \in H_{r}$,  in which $\sum\limits_{i\in V_{S_{r}^h}}{x_i'}- {\sum\limits_{(i, j)\in E_{S_{r}^h}}}{y_{ij}'} < \sum\limits_{i\in V_{T_{r}^h}}{x_i'}- {\sum\limits_{(i, j)\in E_{T_{r}^h}}}{y_{ij}'}$. 
By taking sum over all the components $S_{r}^f \in H_{r}$,  we get the relationship stated in the statement of the Lemma.\qed

\end{proof}


Lemma~\ref{lemma6.3} states that in a non-empty connected component $C$ of $G^*$, if $r=|E_{C}|$, then $C=H_{|E_{C}|}$. Therefore, we can derive the following corollary.

\begin{corollary} \label{cor7}
In a non-empty connected component $C$ of $G^*$,  Let $T$ be the remaining tree derived from $C$ using Algorithm~\ref{algo1} (Steps~6-8). Then, the following relationship holds:

$$\sum\limits_{i\in V_{C}}{x_i^*}-\sum\limits_{(i, j)\in E_{C}}{y_{ij}^*} \geq \sum\limits_{i\in V_{T}}{x_i'}-\sum\limits_{(i, j)\in E_{T}}{y_{ij}'}.$$ 
\end{corollary}


We split the terms on the left-hand side of the inequality (in Collolary above) based on the types of edges in the component $C$. The edge set $E_{C}$ is partitioned into three sets: $E_{T}$ (tree edges), $E_{D}$ (deleted edges), and $E_{NT}$ (non-tree edges). $E_{T}$ can be further divided into two sets: $E_{T^{ \geq .5}}$ = $\{(i, j) \in E_{T} ~|~  y_{ij}^* \geq 0.5\} $ and $E_{T^{< .5}}$ = $\{(i, j)\in E_{T} ~|~  y_{ij}^* < 0.5\}$. In addition, the vertex set $V_{C}$ is partitioned into $V_{T}$ and $V_{D}$. These sets are related as follows:

\begin{multline} \label{eq6}
\sum\limits_{i\in V_{T} }{x^*_{i}} +\sum\limits_{i\in V_{D}}{x_i^*}-\sum\limits_{(i, j)\in E_{{D}}}{y^*_{i, j}}  - \sum\limits_{(i, j)\in E_{T^{ \geq.5}}}{y^*_{i, j}} -\sum\limits_{(i, j)\in E_{T^{ <.5}}}{y^*_{i, j}} - \sum\limits_{(i, j)\in E_{NT}}{y^*_{i, j}}\\
\geq \sum\limits_{i\in V_{T}}{x'_{i}}-\sum\limits_{(i, j)\in E_{T^{ \geq.5}}}{y'_{i, j}}-\sum\limits_{(i, j)\in E_{T^{ <.5}}}{y'_{i, j}}
\end{multline}


We define some constants as follows. For each $(i, j) \in E_{T^{ \geq .5}}$,  $\beta_{ij} = y_{ij}^*-\frac{1}{2}y_{ij}'$ where $0\leq \beta_{ij} \leq \frac{1}{2}$.  For each $(i, j) \in E_{T^{ < .5}}$,  $\zeta_{ij} = \frac{1}{2}y_{ij}'-y_{ij}^*$ where $0\leq \zeta_{ij} \leq \frac{1}{2}$. For each $i \in V_{T}$,  $\gamma_i = x_i^*-\frac{1}{2}x_i'$ where $\frac{-1}{2} \leq \gamma_{i} \leq \frac{1}{2}$.  Using these constants,  \eqref{eq6} can be rewritten as 

\begin{multline}\label{eq7}
  \sum\limits_{i\in V_{D}}{x_i^*}-\sum\limits_{(i, j)\in E_{{D}}}{y^*_{ij}}  + \sum\limits_{i\in V_{T} }[{\frac{1}{2}x'_{i}} + {\gamma_i}] - \sum\limits_{(i, j)\in E_{T^{ \geq.5}}}[\frac{1}{2}{y_{ij}'} + {\beta_{ij}} ]\\
  -\sum\limits_{(i, j)\in E_{T^{ <.5}}}[\frac{1}{2}{y_{ij}'}-{\zeta_{ij}}] - \sum\limits_{(i, j) \in E_{NT}}{y^*_{ij}}  \geq 
 \sum\limits_{i \in V_{T}}{x'_{i}}-\sum\limits_{(i, j) \in E_{T^{ \geq.5}}}{y'_{ij}}-\sum\limits_{(i, j) \in E_{T^{<.5}}}{y'_{ij}}   
\end{multline}


If we subtract $\frac{1}{2}\sum_{V_{T}}x'_i$, $\frac{1}{2}\sum_{E_{T^{ \geq .5}}}{y'_{ij}}$, and $\frac{1}{2}\sum_{E_{T^{ < .5}}}{y'_{ij}}$ from both sides of \eqref{eq7}, we have the following lemma.

\begin{lemma}\label{lemma6.4}
In a non-empty connected component $C$ of $G^*$,  the following  holds: 
\begin{multline*}
\sum\limits_{i\in V_{D}}{x_i^*}-\sum\limits_{(i, j)\in E_{{D}}}{y^*_{ij}}+\sum\limits_{i \in V_{T}}{\gamma_i}
 -\sum\limits_{(i, j) \in E_{T^{ \geq.5}}}{\beta_{ij}} +\sum\limits_{(i, j) \in E_{T^{<.5}}}{\zeta_{ij}} 
  -\sum\limits_{(i, j) \in E_{NT}}{y^*_{ij}}
  \geq \\
\frac{1}{2}\left(
\sum\limits_{i \in V_{T}}{x'_{i}}- \sum\limits_{(i, j) \in E_{T^{ \geq.5}}}{y'_{ij}}-\sum\limits_{(i, j) \in E_{T^{<.5}}}{y'_{ij}}\right)
\end{multline*}

\end{lemma}


Lemma~\ref{lemma6.3} states that in a non-empty connected component $C$, if $H_{r}$ is the sub-graph formed by considering the first $r$ edges of Kruskal's algorithm from $E_{C}$, then the following inequality holds.
$$\sum\limits_{i\in V_{H_{r}}}{x_i^*} - \sum\limits_{(i, j)\in E_{H_{r}}}{y_{ij}^*} \geq \sum\limits_{f=1}^k (\sum\limits_{i\in V_{T_{r}^f}}{x_i'}- {\sum\limits_{(i, j)\in E_{T_{r}^f}}}{y_{ij}'}).$$


It is possible to perform the same partitioning of the edge set and vertex set as was done to derive Lemma~\ref{lemma6.4} for the above relationship.
The edge set $E_{H_{r}}$ is partitioned into three sets: $E_{T_{r}} = \{(i, j) \in E_{H_{r}} |(i, j) \in T\}$, 
$E_{D_{r}} = \{(i, j) \in E_{H_{r}} |(i, j) \in E_{D}\}$,  $E_{NT_{r}} = \{(i, j) \in E_{H_{r}} | (i, j) \in NT\}$. The set $E_{T_{r}}$ is further partitioned into two sets: $E_{T_{r}^{\geq .5}} = \{(i, j) \in E_{T_{r}} |  y_{ij}^* \geq .5\}$ and $E_{T_{r}^{< .5}} = \{(i, j) \in E_{T_{r}} |  y_{ij}^* < .5\}$. Similarly,  the vertex set $V_{H_{r}}$ is partitioned into $V_{T_{r}}$ (the set of all vertices which are endpoints of edges in  $E_{T_{r}}$)  and $V_{D_{r}} = \{ i \in V_{H_{r}} | i \notin V_{T_{r}} \}$.

Note that $E_{H_{r}}=\sum\limits_{f=1}^k E_{S_{r}^f}$,  $\; E_{T_{r}}=\sum\limits_{f=1}^k E_{T_{r}^f}$,  $\; E_{D_{r}}=\sum\limits_{f=1}^k E_{D_{r}^f}$,  $\; E_{NT_{r}}=\sum\limits_{f=1}^k E_{NT_{r}^f}$,  $\; V_{H_{r}}=\sum\limits_{f=1}^k V_{S_{r}^f}$,  $\; V_{T_{r}}=\sum\limits_{f=1}^k V_{T_{r}^f}$,  and $ V_{D_{r}}=\sum\limits_{f=1}^k V_{D_{r}^f}$. 

Then,  the following corollary can be derived. 

\begin{corollary} \label{cor6.5}
In a connected component $C$, if we consider the first $r$ edges from $E_{C}$ using Kruskal's algorithm to form the sub-graph $H_r$, then we have:

\begin{multline*}
\sum\limits_{i \in V_{D_{r}}}{x_i^*}- \sum\limits_{(i, j) \in E_{D_{r}}}y_{ij}^*  +  \sum\limits_{i \in V_{T_{r}}}{\gamma_i} -\sum\limits_{(i, j) \in E_{T_{r}^ {\geq.5}}}{\beta_{ij}} +\sum\limits_{(i, j) \in E_{T_{r}^ {<.5}}}{\zeta_{ij}} -\sum\limits_{(i, j) \in E_{NT_{r}}}{y^*_{ij}} \geq \\
\frac{1}{2}\left(\sum\limits_{i \in V_{T_{r}}}{x_i'} - \sum\limits_{(i, j) \in E_{T_{r}^{ \geq.5}}}{y_{ij}'} - \sum\limits_{(i, j) \in E_{T_{r}^{ < .5}}}{y_{ij}'}\right)
\end{multline*}

\end{corollary}


$T$ is a connected sub-tree of MST $M$ on component $C$. Hence, $|V_{T}| - |E_{T}| = 1$. $E_{T}$ is partitioned into two sets $E_{T^{\geq .5}}$ and $E_{T^{< .5}}$. Therefore, 

$$\sum\limits_{i \in V_{T}}{x_i'}-\sum\limits_{(i, j) \in E_{T^{ \geq .5}}}{y_{ij}'}-\sum\limits_{(i, j) \in E_{T^{<.5}}}{y_{ij}'}=1.$$ 
Then,  from Lemma~\ref{lemma6.4} we have

\begin{multline*}
\sum\limits_{i \in V_{D}}{x_i^*} - \sum\limits_{(i, j) \in E_{D}}{y_{ij}^*} + \sum\limits_{i \in V_{T}}{\gamma_i}-\sum\limits_{(i, j) \in E_{T^ {\geq .5}}}{\beta_{ij}}+\sum\limits_{(i, j) \in E_{T^{<.5}}}{\zeta_{ij}}-\sum\limits_{(i, j) \in E_{NT}}{y_{ij}^*} \geq\frac{1}{2}
\end{multline*}

Since $0 \le (1-w_{ij}) \leq 1$,  for each $(i, j) \in E_{D}$ multiplying $(1-w_{ij})$ with the corresponding $y_{ij}^*$ for each $(i, j) \in E_{D}$ in the above equation yields the following inequality:

\begin{multline}{\label{eq9}}
\sum\limits_{i \in V_{D}}{x_i^*} - \sum\limits_{(i, j) \in E_{D}}{y_{ij}^*}(1-w_{ij}) + \sum\limits_{i \in V_{T}}{\gamma_i}-\sum\limits_{(i, j) \in E_{T^ {\geq.5}}}{\beta_{ij}}+\sum\limits_{(i, j) \in E_{T^{<.5}}}{\zeta_{ij}}\\
-\sum\limits_{(i, j) \in E_{NT}}{y_{ij}^*} \geq\frac{1}{2}
\end{multline}

\textcolor{black}{The next lemma states that if the terms involving edges in LHS of the inequality stated in Lemma~\ref{lemma6.4} are multiplied by (1-$w_{ij}$), then the modified LHS remains greater than zero. }


\begin{lemma}
    
\label{lemma6.6}
In a non-empty connected component $C$ of $G^*$,  the following holds:
\begin{multline}\label{eq8}
\sum\limits_{i \in V_{D}}{x_i^*} - \sum\limits_{(i, j) \in E_{D}}{y_{ij}^*}(1-w_{ij}) + \sum\limits_{i \in V_{T}}{\gamma_i} -\sum\limits_{(i, j) \in E_{T^{\geq.5}}}{\beta_{ij}}(1-w_{ij})\\ +\sum\limits_{(i, j) \in E_{T ^{<.5}}}{\zeta_{ij}}(1-w_{ij}) -\sum\limits_{(i, j) \in E_{NT}}{y^*_{ij}}(1-w_{ij}) \geq 0 \end{multline}
\end{lemma}


LHS of \eqref{eq8} in Lemma~\ref{lemma6.6} is obtained by adding 
$$\sum\limits_{(i, j) \in E_{T^{ \geq .5}}}{\beta_{ij}w_{ij}} -  \sum\limits_{(i, j) \in E_{T^{<.5}}}{\zeta_{ij}w_{ij}} +  \sum\limits_{(i, j) \in E_{NT}}{y_{ij}^*w_{ij}}$$
to the LHS of \eqref{eq9}.


We refer to the terms $\beta_{ij}w_{ij}$, $\zeta_{ij}w_{ij}$, and $y_{ij}^*w_{ij}$ as edge-weight-multiplied terms in \eqref{eq9}.  We classify $\zeta_{ij}w_{i,j}$ as the negative additive term and the other two as positive. Using \eqref{eq9} it is easy to note that if $\sum\limits_{(i, j) \in E_{T \geq .5}}{\beta_{ij}w_{ij}} - \sum\limits_{(i, j) \in E_{T^{<.5}}}{\zeta_{ij}w_{ij}} + \sum\limits_{(i, j) \in E_{NT}}{y_{ij}^*w_{ij}} \geq 0$ then \eqref{eq8} is true. But when $\sum\limits_{(i, j) \in E_{T ^{\geq .5}}}{\beta_{ij}w_{ij}} - \sum\limits_{(i, j) \in E_{T^{<.5}}}{\zeta_{ij}w_{ij}} + \sum\limits_{(i, j) \in E_{NT}}{y_{ij}^*w_{ij}} < 0$, then it is not straightforward.


 \textcolor{black}{To prove Lemma~\ref{lemma6.6}, we partition the edge set $E_C$ into two non-trivial partitions, the procedure for which is discussed below. We then capture the relationships within each partition using a few lemmas. We evaluate the summation of corresponding terms for the edges and vertices of both partitions separately. Finally, we add the derived inequalities to prove the Lemma.}\\


Recall that $SoE_{C}=(e_1, e_2, \cdots, e_{|E_{C}|})$ is defined as the sequence of all edges in $E_{C}$ taken in the same order in which Kruskal's algorithm considered edges while forming the MST $M$ on $C$. Each edge is present in exactly one of the edge partitions $E_{T^{ \geq .5}}$, $E_{T^{ < .5}}$, $E_{D}$ or $E_{NT}$ of $E_{C}$. We break $SoE_{C}$ into $k$ sub-sequences $SoE_{C}^{(1, r_1)}, SoE_{C}^{(r_1+1, r_2)}, \cdots, SoE_{C}^{(r_{k-1}+1, |E_{C}|)}$, where $SoE_{C}^{(r_z+1, r_{(z+1)})}$ denotes the sub-sequence $(e_{r_z+1}, e_{r_z+2}, \cdots, e_{r_{(z+1)}})$.

We refer to the sum of edge-weight-multiplied terms corresponding to the edges in a sub-sequence as the \textit{sum of weighted terms}. For each $(i, j) \in E_{{T_{t}^{<.5}}}$, we add $-\zeta_{ij}w_{ij}$ to the sum of weighted terms. For each $(i, j) \in E_{{T_{t}^{\geq .5}}}$, we add $\beta_{ij}w_{ij}$ to the sum of weighted terms. For each $(i, j) \in E_{{NT_{t }}}$, we add $y_{ij}^*w_{ij}$ to the sum of weighted terms. 
Our idea of breaking $SoE_{C}$ is to create sub-sequences that facilitate the computation of the sum of weighted terms for each edge partition is as follow.


\begin{itemize}
    \item Step1: The set $SoE_{C}^{(1, r_1)}$ contains all edges that appeared before the arrival of the first edge from set $E_{T^{ < .5}}$. The sum of the edge-weight-multiplied terms corresponding to the edges in the sub-sequence $SoE_{C}^{(1, r_1)}$ is positive as only edges from the sets $E_{T^{ \geq .5}}$ or $E_{NT}$ have appeared so far and the corresponding edge-weight-multiplied terms are positive additive terms.

    \item Step 2: The first edge in $SoE_{C}^{(r_1+1,  r_2)}$ is from set $E_{T^{ < .5}}$, and therefore, the corresponding edge-weight-multiplied term is a negative additive term. Hence,  the sum of weighted terms for this sub-sequence so far is also negative. We keep adding more edges of any type to $SoE_{C}$ until the sum of weighted terms for this sub-sequence becomes non-negative for the first time. At this point,  if the upcoming edges are from the sets $E_{T^{ \geq .5}}$,  or $E_{NT}$,  we add them to the sequence until the arrival of the next edge from set $E_{T^{ < .5}}$ til the sum of weighted terms for this sub-sequence became non-negative for the first time. The edges added so far form a sub-sequence $SoE_{C}^{(r_1+1,  r_2)}$ where the last edge $e_{r_2}$ is definitely from $E_{T^{ \geq .5}}$,  or $E_{NT}$. Note that the sum of weighted terms,  i.e.,  the sum of edge-weight-multiplied terms corresponding to the edges in the sub-sequence $SoE_{C}^{(r_1+1,  r_2)}$ is positive.
\end{itemize}

We repeat Step 2 to form the remaining sub-sequences, ensuring the sum of weighted terms of all sub-sequences except the last one is positive. If 
$$\sum\limits_{(i, j) \in E_{T^{\geq .5}}}{\beta_{ij}w_{ij}} - \sum\limits_{(i, j) \in E_{T^{<.5}}}{\zeta_{ij}w_{ij}} + \sum\limits_{(i, j) \in E_{NT}}{y_{ij}^*w_{ij}} < 0,$$  
then the sum of weighted terms of the last sub-sequence is negative.  

Let $SoE_{C}$ be partitioned into $p$ sub-sequences as above. Let $e_t=e_{r_{(p-1)}}$ be the last edge of $(p-1)$th sub-sequence. Let $e_{t+1}=(u, v)$ be the first edge in the last sub-sequence $SoE_{C}^{(r_{p-1}+1,  |E_{C}|)}$ which is the $(t+1)^{th}$ edge,  i.e.,  $(r_{(p-1)}+1)^{th}$ edge in sequence $SoE_{C}$. Let $E_{H_{t}}$ be the set of the first $t$ edges in $SoE_{C}$. Let $E_{\overline{H_{t}}}  = E_{C} \setminus E_{H_{t}}$ be the set of all edges considered since consideration of $(u, v)$ by Kruskal's algorithm from $E_{C}$. $E_{H_{t}}=\{e_1,  e_2, \cdots, e_t\}$ and $E_{\overline{H_{t}}}=\{e_{t+1},  e_{t+2}, \cdots,  e_{|E_{C}|}\}$. Let $V_{H_{t}}$ be the set of all vertices corresponding to $E_{H_{t}}$ appearing before $(u, v)$. Let $V_{\overline{H_{t}}}$ be the set of all vertices which are endpoints of edges in $E_{\overline{H_{t}}}$ and have not appeared in $V_{H_{t}}$. The edge set $E_{H_{t}}$ is partitioned into: $E_{D_{t}}$,  $E_{T_{t} ^ {\geq .5}}$,  $E_{T_{t}^ {<.5}}$,  and $E_{NT_{t}}$,  similar to the partitioning done earlier. Note that $E_{T_{t} ^ {<.5}}\cup E_{T_{t} ^ {\geq.5}}$ is collectively referred as $E_{T_{t}}$. The edge set $E_{\overline{H_{t}}}$ is partitioned into: $E_{\overline{D_{t}}}$,  $E_{\overline{T_{t} ^ {\geq .5}}}$,  $E_{\overline{T_{t} ^ {<.5}}}$,  and $E_{\overline{NT_{t}}}$. The vertex set $V_{H_{t}}$ is partitioned into: $V_{D_{t}}$ and $V_{T_{t}}$. Further,  the set $V_{T_{t}}$ is partitioned into $V_{T_{t}^ {\geq .5}}$ and $V_{T_{t}^ { < .5}}$. The vertex set $V_{\overline{{H_{t}}}}$ is partitioned into: $V_{\overline{{D_{t}}}}$ and $V_{\overline{{T_{t}}}}$. Further,  the set $V_{\overline{{T_{t}}}}$ is partitioned into $V_{\overline{{T_{t}^ {\geq .5}}}}$ and $V_{\overline{{T_{t}^ { < .5}}}}$.

\begin{lemma}
Let $e_t$ be the last edge of $(p-1)$th sub-sequence. Let $H_{t}$ be a subgraph of a non-empty component $C$ constructed based on edge set $E_{H_{t}}$ which is the set of the first $t$ edges in $SoE_{C}$, i.e, $E_{H_{t}}=\{e_1,  e_2, \cdots, e_t\}$. Let ${H_t}$ has $Q$ connected components. Then, the following relationship holds:
\begin{multline}\label{eq10}
\sum\limits_{i \in V_{D_{t}}}{x_i^*} - \sum\limits_{(i, j) \in E_{D_{t}}}{y_{ij}^*}(1-w_{ij}) + \sum\limits_{i \in V_{T_{t}}}{\gamma_i} -\sum\limits_{(i, j) \in E_{T_{t} ^ {\geq.5}}}{\beta_{ij}}(1-w_{ij})\\ +\sum\limits_{(i, j) \in E_{T_{t}^ {<.5}}}{\zeta_{ij}(1-w_{ij})}  - \sum\limits_{(i, j) \in E_{NT_{t}}}{y^*_{ij}(1-w_{ij})} \geq \frac{Q}{2}
\end{multline}
\end{lemma}
\begin{proof}

Consider $H_{t}$ which is formed by taking the set of the first $t$ edges in $SoE_{C}$ based on component $C$, from Corollary~\ref{cor6.5}, we have: 
\begin{multline*}
   \sum\limits_{i \in V_{D_{t}}}{x_i^*} - \sum\limits_{(i, j) \in E_{D_{t}}}{y_{ij}^*} + \sum\limits_{i \in V_{T_{t}}}{\gamma_i} -\sum\limits_{(i, j) \in E_{T_{t} ^ {\geq.5}}}{\beta_{ij}} +\sum\limits_{(i, j) \in E_{T_{t}^ { <.5}}}{\zeta_{ij}} -\sum\limits_{(i, j) \in E_{NT_{t}}}{y^*_{ij}} \geq \\
   \frac{1}{2}\left(\sum\limits_{i \in V_{T_{t}}}{x_i'} - \sum\limits_{(i, j) \in E_{T_{t}^{ \geq.5}}}{y_{ij}'} - \sum\limits_{(i, j) \in E_{T_{t}^ { < .5}}}{y_{ij}'}\right)
\end{multline*}

 Note that $(V_{T_{t}}, E_{T_{t}} )$ has $Q$ connected components because removal of $V_{D_{t}}$,  $E_{D_{t}}$,  and $E_{NT_{t}}$ from $H_{t}$ results in $(V_{T_{t}}, E_{T_{t}} )$ without any change in the number of connected components. Therefore,  $\sum\limits_{i \in V_{T_{t}}}{x_i'} - \sum\limits_{(i, j) \in E_{T_{t}^ { \geq.5}}}{y_{ij}'} - \sum\limits_{(i, j) \in E_{T_{t}^ { < .5}}}{y_{ij}'} \geq Q$. Hence, 

 \begin{multline*}
 \sum\limits_{i \in V_{D_{t}}}{x_i^*} - \sum\limits_{(i, j) \in E_{D_{t}}}{y_{ij}^*} + \sum\limits_{i \in V_{T_{t}}}{\gamma_i} -\sum\limits_{(i, j) \in E_{T_{t} ^  {\geq.5}}}{\beta_{ij}} +\sum\limits_{(i, j) \in E_{T_{t}^{ <.5}}}{\zeta_{ij}} -\sum\limits_{(i, j) \in E_{NT_{t}}}{y^*_{ij}} \geq \frac{Q}{2}.   
 \end{multline*}

 Since the  sum of weighted terms for all the sub-sequences except for the last sub-sequence occurring after edge $e_t$,  i.e.,   $\sum\limits_{(i, j) \in E_{T_{t}^{ \geq .5}}}{\beta_{ij}w_{ij}} - \sum\limits_{(i, j) \in E_{T_{t}^{<.5}}}{\zeta_{ij}w_{ij}} + \sum\limits_{(i, j) \in E_{NT_{t}}}{y_{ij}^*w_{ij}}$ is positive,  adding it to the left hand side of the above equation give us:
\begin{multline*}
\sum\limits_{i \in V_{D_{t}}}{x_i^*} - \sum\limits_{(i, j) \in E_{D_{t}}}{y_{ij}^*} + \sum\limits_{i \in V_{T_{t}}}{\gamma_i} -\sum\limits_{(i, j) \in E_{T_{t} ^ {\geq.5}}}{\beta_{ij}}(1-w_{ij})\\ +\sum\limits_{(i, j) \in E_{T_{t}^ {<.5}}}{\zeta_{ij}(1-w_{ij})}  - \sum\limits_{(i, j) \in E_{NT_{t}}}{y^*_{ij}(1-w_{ij})} \geq \frac{Q}{2}
\end{multline*}
For  each $(i, j) \in E_{D_{t}}$,  $(1-w_{ij}) \leq 1$,  after multiplying corresponding $(1-w_{ij})$ with $y_{ij}^*$ we get \eqref{eq10}.\qed
\end{proof}

Next,  we analyse the edges in $E_{\overline{H_{t}}}=\{e_{t+1},  e_{t+2}, \cdots,  e_{|E_{C}|}\}$. The sum of weighted terms of the last sub-sequence containing these edges is negative. We can express it  as:
\begin{equation}\label{eq11}
\sum\limits_{(i, j) \in E_{\overline{T_{t}^{<.5}}}}{\zeta_{ij}w_{ij}} \geq \sum\limits_{(i, j) \in E_{\overline{T_{t}^{\geq .5}}}}{\beta_{ij}w_{ij}} + \sum\limits_{(i, j) \in E_{\overline{NT_{t }}}}{y_{ij}^*}w_{ij}
\end{equation}

Recall,  each $(i, j) \in E_{\overline{T_{t}^{<.5}}}$ adds $-\zeta_{ij}w_{ij}$,  each $(i, j) \in E_{\overline{T_{t}^{\geq .5}}}$,  adds $\beta_{ij}w_{ij}$,  and  each $(i, j) \in E_{\overline{NT_{t }}}$ adds $y_{ij}^*w_{ij}$ to the sum of weighted terms.  The first edge in the last sub-sequence $e_{t+1}$ belongs to $E_{\overline{T_{t}^{<.5}}}$.   Hence,  the sum of weighted terms for the last sub-sequence is negative. The sum of weighted terms in the last sub-sequence never becomes positive and remains negative as each subsequent edge is considered, starting from the appearance of the first edge $e_{t+1}$ in the sub-sequence. This is true even if the edges considered after $e_{t+1}$ belong to $E_{\overline{T_{t}^{\geq .5}}}$ or $E_{\overline{NT_{t}}}$, which add positive edge-weight-multiplied terms to the sum of weighted terms for this sub-sequence.

\textcolor{black}{Next, we prove the following lemma, which states that the inequality mentioned in \eqref{eq11} holds even when weight coefficients are ignored.}\\

\begin{lemma}\label{lemma_claim0} The following relationship holds among the edges from $E_{\overline{T_{t}}}$
\begin{equation}\label{eq12}
\sum\limits_{(i, j) \in E_{\overline{T_{t}^{<.5}}}}{\zeta_{ij}} \geq \sum\limits_{(i, j) \in E_{\overline{T_{t}^{\geq .5}}}}{\beta_{ij}} + \sum\limits_{(i, j) \in E_{\overline{NT_{t }}}}{y_{ij}^*}
\end{equation}
\end{lemma}

\begin{proof} [Proof Idea:]
    
 The first edge in the last sub-sequence is $e_{t+1}=(u_1, v_1)$ which belongs in $E_{\overline{{T_{t} ^{ < .5}}}}$,  therefore,  the sum of weighted terms for the last sub-sequence is initialized with $-\zeta_{u_1v_1}w_{u_1v_1}$. Let $e_{t+s} = (u_2, v_2)$ be the second edge of type $E_{\overline{{T_{t} ^{ < .5}}}}$ in the sub-sequence $SoE_{C}^{(r_{k-1}+1,  |E_{C}|)}$.  
Then,  the weighted sum corresponding to the edges of type $E_{\overline{{T_{t}^{\geq .5}}}}$ and $E_{\overline{{NT_{t}}}}$ occurring between $e_{t+1}$ and $e_{t+s}$ is

$\sum\limits_{ (i, j)\in\{e_{t+k}\in E_{\overline{T_{t}^{\geq .5}}}\; |\; 1<k<s\}}{\beta_{ij}w_{ij}} + \sum\limits_{ (i, j)\in\{e_{t+k}\in E_{\overline{NT_{t}}}\; |\; 1<k<s\}}{y_{ij}^*w_{ij}}$. Recall the sum of weighted terms for the last sub-sequence never becomes non-negative and remains negative at each step of considering the next edges since consideration of $e_{t+1}$. Therefore,  
 $$\zeta_{u_1v_1}w_{u_1v_1} \geq \sum\limits_{ (i, j)\in\{e_{t+k}\in E_{\overline{T_{t}^{\geq .5}}}\; |\; 1<k<s\}}{\beta_{ij}w_{ij}} + \sum\limits_{ (i, j)\in\{e_{t+k}\in E_{\overline{NT_{t}}}\; |\; 1<k<s\}}{y_{ij}^*w_{ij}}.$$
 
 
Since,   $w_{e_{t+1}} \leq w_{e_{t+2}} \leq \cdots \leq w_{e_{t+s}}$,  we have
 
 \begin{equation}\label{eq13}
 \zeta_{u_1v_1} \geq \sum\limits_{ (i, j)\in\{e_{t+k}\in E_{\overline{T_{t}^{\geq .5}}}\; |\; 1<k<s\}}{\beta_{ij}} + \sum\limits_{ (i, j)\in\{e_{t+k}\in E_{\overline{NT_{t}}}\; |\; 1<k<s\}}{y_{ij}^*}
  \end{equation}
  
As $w_{u_1v_1} \leq w_{u_2v_2}$,  hence from the the above equation,  we have 
\begin{multline*}
\zeta_{u_1v_1}w_{u_1v_1} - \sum\limits_{ (i, j)\in\{e_{t+k}\in E_{\overline{T_{t}^{\geq .5}}}\; |\; 1<k<s\}}{\beta_{ij}w_{u_1v_1}} - \sum\limits_{ (i, j)\in\{e_{t+k}\in E_{\overline{NT_{t}}}\; |\; 1<k<s\}}{y_{ij}^*w_{u_1v_1}} \\ \leq \zeta_{u_1v_1}w_{u_2v_2} - \sum\limits_{ (i, j)\in\{e_{t+k}\in E_{\overline{T_{t}^{\geq .5}}}\; |\; 1<k<s\}}{\beta_{ij}w_{u_2v_2}} - \sum\limits_{ (i, j)\in\{e_{t+k}\in E_{\overline{NT_{t}}}\; |\; 1<k<s\}}{y_{ij}^*w_{u_2v_2}}
\end{multline*}

Since,   $w_{u_1v_1}=w_{e_{t+1}} \leq w_{e_{t+2}} \leq \cdots \leq w_{e_{t+s}}=w_{u_2v_2}$,  We can rewrite the above equations as

\begin{multline} \label{eq20}
\zeta_{u_1v_1}w_{u_1v_1} - \sum\limits_{ (i, j)\in\{e_{t+k}\in E_{\overline{T_{t}^{\geq .5}}}\; |\; 1<k<s\}}{\beta_{ij}w_{ij}} - \sum\limits_{ (i, j)\in\{e_{t+k}\in E_{\overline{NT_{t}}}\; |\; 1<k<s\}}{y_{ij}^*w_{ij}} \\ \leq \zeta_{u_1v_1}w_{u_2v_2} - \sum\limits_{ (i, j)\in\{e_{t+k}\in E_{\overline{T_{t}^{\geq .5}}}\; |\; 1<k<s\}}{\beta_{ij}w_{u_2v_2}} - \sum\limits_{ (i, j)\in\{e_{t+k}\in E_{\overline{NT_{t}}}\; |\; 1<k<s\}}{y_{ij}^*w_{u_2v_2}}
\end{multline}

Let $e_{t+z} = (u_3, v_3)$  be the third edge of type $E_{\overline{T_{t}^{< .5}}}$ in the sequence $SoE_{C}^{(r_{k-1}+1,  |E_{C}|)}$. The sum of weighted terms for the last sub-sequence just before consideration of $e_{t+z}=(u_3,  v_3)$ is $-\zeta_{u_1v_1} w_{u_1v_1}- \zeta_{u_2v_2} w_{u_2v_2} + \sum\limits_{ (i, j)\in\{e_{t+k}\in E_{\overline{T_{t}^{\geq .5}}}\; |\; 1<k<z\}}{\beta_{ij}w_{ij}} + \sum\limits_{ (i, j)\in\{e_{t+k}\in E_{\overline{NT_{t}}}\; |\; 1<k<z\}}{y_{ij}^*w_{ij}}$. Since it is negative,  we have the following

$$
 \zeta_{u_1v_1} w_{u_1v_1}+ \zeta_{u_2v_2} w_{u_2v_2} \geq \sum\limits_{ (i, j)\in\{e_{t+k}\in E_{\overline{T_{t}^{\geq .5}}}\; |\; 1<k<z\}}{\beta_{ij}w_{ij}} + \sum\limits_{ (i, j)\in\{e_{t+k}\in E_{\overline{NT_{t}}}\; |\; 1<k<z\}}{y_{ij}^*w_{ij}}.
$$

If we regroup the positive edge-weight-multiplied terms for the edges in the sub-sequence $\{e_{t+2},  \cdots,  e_{t+s-1}\}$ and $\{e_{t+s+1},  \cdots,  e_{t+z-1}\}$ separately,  we have
\begin{multline*}
 \zeta_{u_1v_1} w_{u_1v_1}+ \zeta_{u_2v_2} w_{u_2v_2} \geq \sum\limits_{ (i, j)\in\{e_{t+k}\in E_{\overline{T_{t}^{\geq .5}}}\; |\; 1<k<s\}}{\beta_{ij}w_{ij}} + \sum\limits_{ (i, j)\in\{e_{t+k}\in E_{\overline{NT_{t}}}\; |\; 1<k<s\}}{y_{ij}^*w_{ij}}\\ + \sum\limits_{ (i, j)\in\{e_{t+k}\in E_{\overline{T_{t}^{\geq .5}}}\; |\; s<k<z\}}{\beta_{ij}w_{ij}} + \sum\limits_{ (i, j)\in\{e_{t+k}\in E_{\overline{NT_{t}}}\; |\; s<k<z\}}{y_{ij}^*w_{ij}}.$$
\end{multline*}

 Using \eqref{eq20},  we may rewrite the above equation as  
 
 \begin{multline*}
 \zeta_{u_1v_1} w_{u_2v_2}+ \zeta_{u_2v_2} w_{u_2v_2} \geq \sum\limits_{ (i, j)\in\{e_{t+k}\in E_{\overline{T_{t}^{\geq .5}}}\; |\; 1<k<s\}}{\beta_{ij}w_{u_2v_2}} + \sum\limits_{ (i, j)\in\{e_{t+k}\in E_{\overline{NT_{t}}}\; |\; 1<k<s\}}{y_{ij}^*w_{u_2v_2}}\\ + \sum\limits_{ (i, j)\in\{e_{t+k}\in E_{\overline{T_{t}^{\geq .5}}}\; |\; s<k<z\}}{\beta_{ij}w_{ij}} + \sum\limits_{ (i, j)\in\{e_{t+k}\in E_{\overline{NT_{t}}}\; |\; s<k<z\}}{y_{ij}^*w_{ij}}.$$
\end{multline*}
 
Dividing both sides by $w_{u_2v_2}$,  we get
 
  \begin{multline*}
 \zeta_{u_1v_1} + \zeta_{u_2v_2} \geq \sum\limits_{ (i, j)\in\{e_{t+k}\in E_{\overline{T_{t}^{\geq .5}}}\; |\; 1<k<s\}}{\beta_{ij}} + \sum\limits_{ (i, j)\in\{e_{t+k}\in E_{\overline{NT_{t}}}\; |\; 1<k<s\}}{y_{ij}^*}\\ + \sum\limits_{ (i, j)\in\{e_{t+k}\in E_{\overline{T_{t}^{\geq .5}}}\; |\; s<k<z\}}{\beta_{ij}\frac{w_{ij}}{w_{u_2v_2}}} + \sum\limits_{ (i, j)\in\{e_{t+k}\in E_{\overline{NT_{t}}}\; |\; s<k<z\}}{y_{ij}^*\frac{w_{ij}}{w_{u_2v_2}}}.$$
\end{multline*}

 Since,   $w_{u_2v_2}=w_{e_{t+s}} \leq w_{e_{t+s+1}} \leq \cdots \leq w_{e_{t+z}}=w_{u_3v_3}$,  we have
\begin{multline*}\label{eq16}
   \zeta_{u_1v_1} + \zeta_{u_2v_2} \geq \sum\limits_{ (i, j)\in\{e_{t+k}\in E_{\overline{T_{t}^{\geq .5}}}\; |\; 1<k<s\}}{\beta_{ij}} + \sum\limits_{ (i, j)\in\{e_{t+k}\in E_{\overline{NT_{t}}}\; |\; 1<k<s\}}{y_{ij}^*}\\ + \sum\limits_{ (i, j)\in\{e_{t+k}\in E_{\overline{T_{t}^{\geq .5}}}\; |\; s<k<z\}}{\beta_{ij}} + \sum\limits_{ (i, j)\in\{e_{t+k}\in E_{\overline{NT_{t}}}\; |\; s<k<z\}}{y_{ij}^*}.
\end{multline*}

If we merge the positive edge-weight-multiplied terms for the edges in the sub-sequence 

$\{e_{t+2},  \cdots,  e_{t+s-1}\}$ and $\{e_{t+s+1},  \cdots,  e_{t+z-1}\}$,  we have

\begin{multline}
  \zeta_{u_1v_1} + \zeta_{u_2v_2} \geq \sum\limits_{ (i, j)\in\{e_{t+k}\in E_{\overline{T_{t}^{\geq .5}}}\; |\; 1<k<z\}}{\beta_{ij}} + \sum\limits_{ (i, j)\in\{e_{t+k}\in E_{\overline{NT_{t}}}\; |\; 1<k<z\}}{y_{ij}^*}
\end{multline}

We can now prove Lemma~\ref{lemma_claim0} by induction.\qed
\end{proof}

Using similar arguments,  we prove the next lemma.\\

\begin{lemma}\label{lemma_claim} The following relationship holds among the edges from $E_{\overline{T_{t}}}$:
\begin{equation}\label{eq22}
\sum\limits_{(i, j) \in E_{\overline{T_{t}^{<.5}}}}\zeta_{ij}(1-w_{ij}) \geq \sum\limits_{(i, j) \in E_{\overline{T_{t}^{\geq .5}}}}{\beta_{ij}}(1-w_{ij}) + \sum\limits_{(i, j) \in E_{\overline{NT_{t }}}}{y_{ij}^*(1-w_{ij})}
\end{equation} 
\end{lemma}

\begin{proof}[Proof Idea:]
To establish the base case, we multiply with $(1-w_{u_1v_1})$ on both sides of \eqref{eq13}. This gives the following inequality,
\begin{multline*}
 \zeta_{u_1v_1}(1- w_{u_1v_1}) \geq \sum\limits_{ (i, j)\in\{e_{t+k}\in E_{\overline{T_{t}^{\geq .5}}}\; |\; 1<k<s\}}{\beta_{ij} (1- w_{u_1v_1})} + \sum\limits_{ (i, j)\in\{e_{t+k}\in E_{\overline{NT_{t}}}\; |\; 1<k<s\}}{y_{ij}^* (1- w_{u_1v_1})}
 \end{multline*}

Since $w_{u_1v_1} = w_{e_{t+1}} \leq w_{e_{t+2}} \leq \cdots w_{e_{t+s-1}}  \leq w_{e_{t+s}}= w_{u_2v_2} \implies (1-w_{u_1v_1}) = (1-w_{e_{t+1}}) \geq (1-w_{e_{t+2}}) \geq \cdots (1-w_{e_{t+s-1}}) \geq (1-w_{u_2v_2})= (1-w_{e_{t+s}}).$ Hence,  $(1-w_{u_1v_1}) \geq (1-w_{ij})$ for $(i, j)\in\{e_{t+k}\in E_{\overline{T_{t}^{\geq .5}}}\; |\; 1<k<s\}$. Therefore,  we can rewrite the above equation as

\begin{multline*}
 \zeta_{u_1v_1}(1- w_{u_1v_1}) \geq \sum\limits_{ (i, j)\in\{e_{t+k}\in E_{\overline{T_{t}^{\geq .5}}}\; |\; 1<k<s\}}{\beta_{ij} (1- w_{ij})} + \sum\limits_{ (i, j)\in\{e_{t+k}\in E_{\overline{NT_{t}}}\; |\; 1<k<s\}}{y_{ij}^* (1- w_{ij})}
 \end{multline*}

Similarly, if we multiply $(1-w_{u_2v_2})$ on both sides of \eqref{eq16}, we get the following.

\begin{multline}\label{eq18}
   \zeta_{u_1v_1}(1-w_{u_2v_2}) + \zeta_{u_2v_2}(1-w_{u_2v_2}) \geq \sum\limits_{ (i, j)\in\{e_{t+k}\in E_{\overline{T_{t}^{\geq .5}}}\; |\; 1<k<s\}}{\beta_{ij}(1-w_{u_2v_2})} + \\ \sum\limits_{ (i, j)\in\{e_{t+k}\in E_{\overline{NT_{t}}}\; |\; 1<k<s\}}{y_{ij}^*(1-w_{u_2v_2})} + \sum\limits_{ (i, j)\in\{e_{t+k}\in E_{\overline{T_{t}^{\geq .5}}}\; |\; s<k<z\}}{\beta_{ij}(1-w_{u_2v_2})}\\ + \sum\limits_{ (i, j)\in\{e_{t+k}\in E_{\overline{NT_{t}}}\; |\; s<k<z\}}{y_{ij}^*(1-w_{u_2v_2})}.
\end{multline}


Since $w_{u_1v_1} \leq w_{u_2v_2}$ $\implies$   $(1- w_{u_1v_1}) \geq (1-w_{u_2v_2})$,  from \eqref{eq13},  we have 
\begin{multline*}
\footnotesize
\implies \zeta_{u_1v_1}(1-w_{u_1v_1}) - \sum\limits_{ (i, j)\in\{e_{t+k}\in E_{\overline{T_{t}^{\geq .5}}}\; |\; 1<k<s\}}{\beta_{ij}} (1-w_{u_1v_1}) \\ -\sum\limits_{ (i, j)\in\{e_{t+k}\in E_{\overline{NT_{t}}}\; |\; 1<k<s\}}{y_{ij}^*}(1-w_{u_1v_1}) \geq 
\zeta_{u_1v_1}(1-w_{u_2v_2}) \\
-\sum\limits_{ (i, j)\in\{e_{t+k}\in E_{\overline{T_{t}^{\geq .5}}}\; |\; 1<k<s\}}{\beta_{ij}}(1-w_{u_2v_2}) - \sum\limits_{ (i, j)\in\{e_{t+k}\in E_{\overline{NT_{t}}}\; |\; 1<k<s\}}{y_{ij}^*}(1-w_{u_2v_2})
\end{multline*}

Using the inequality above, \eqref{eq18} is rewritten as

\begin{multline}
   \zeta_{u_1v_1}(1-w_{u_1v_1}) + \zeta_{u_2v_2}(1-w_{u_2v_2}) \geq \sum\limits_{ (i, j)\in\{e_{t+k}\in E_{\overline{T_{t}^{\geq .5}}}\; |\; 1<k<s\}}{\beta_{ij}(1-w_{u_1v_1})} + \\ \sum\limits_{ (i, j)\in\{e_{t+k}\in E_{\overline{NT_{t}}}\; |\; 1<k<s\}}{y_{ij}^*(1-w_{u_1v_1})} + \sum\limits_{ (i, j)\in\{e_{t+k}\in E_{\overline{T_{t}^{\geq .5}}}\; |\; s<k<z\}}{\beta_{ij}(1-w_{u_2v_2})}\\ + \sum\limits_{ (i, j)\in\{e_{t+k}\in E_{\overline{NT_{t}}}\; |\; s<k<z\}}{y_{ij}^*(1-w_{u_2v_2})}.
\end{multline}

Recall, $w_{u_1v_1} = w_{e_{t+1}} \leq w_{e_{t+2}} \leq \cdots w_{e_{t+s-1}}  \leq w_{u_2v_2}= w_{e_{t+s}} \leq w_{e_{t+s+1}} \leq \cdots \leq w_{e_{t+z-1}}$ $\implies$ $(1-w_{u_1v_1}) = (1-w_{e_{t+1}}) \geq (1-w_{e_{t+2}}) \geq \cdots (1-w_{e_{t+s-1}} \geq (1-w_{u_2v_2})= (1-w_{e_{t+s}}) \geq (1-w_{e_{t+s+1}}) \geq \cdots \geq (1-w_{e_{t+z-1}})$). Hence,  $(1-w_{u_1v_1}) \geq (1-w_{ij})$ for $(i, j)\in\{e_{t+k}\in E_{\overline{T_{t}^{\geq .5}}}\; |\; 1<k<s\}$ and $(1-w_{u_2v_2}) \geq (1-w_{ij})$ for $(i, j)\in\{e_{t+k}\in E_{\overline{T_{t}^{\geq .5}}}\; |\; s<k<z\}$. Therefore,  we can rewrite as

\begin{multline*}
   \zeta_{u_1v_1}(1-w_{u_1v_1}) + \zeta_{u_2v_2}(1-w_{u_2v_2}) \geq \sum\limits_{ (i, j)\in\{e_{t+k}\in E_{\overline{T_{t}^{\geq .5}}}\; |\; 1<k<s\}}{\beta_{ij}(1-w_{ij})} + \\ \sum\limits_{ (i, j)\in\{e_{t+k}\in E_{\overline{NT_{t}}}\; |\; 1<k<s\}}{y_{ij}^*(1-w_{ij})} + \sum\limits_{ (i, j)\in\{e_{t+k}\in E_{\overline{T_{t}^{\geq .5}}}\; |\; s<k<z\}}{\beta_{ij}(1-w_{ij})}\\ + \sum\limits_{ (i, j)\in\{e_{t+k}\in E_{\overline{NT_{t}}}\; |\; s<k<z\}}{y_{ij}^*(1-w_{ij})}.
\end{multline*}

If we regroup the positive edge-weight-multiplied terms for the edges in the sub-sequence

$\{e_{t+2},  \cdots,  e_{t+s-1}\}$ and $\{e_{t+s+1},  \cdots,  e_{t+z-1}\}$ separately,  we have

\begin{multline*}
   \zeta_{u_1v_1}(1-w_{u_1v_1}) + \zeta_{u_2v_2}(1-w_{u_2v_2}) \geq \sum\limits_{ (i, j)\in\{e_{t+k}\in E_{\overline{T_{t}^{\geq .5}}}\; |\; 1<k<z\}}{\beta_{ij}(1-w_{ij})} + \\ \sum\limits_{ (i, j)\in\{e_{t+k}\in E_{\overline{NT_{t}}}\; |\; 1<k<z\}}{y_{ij}^*(1-w_{ij})}.
\end{multline*}

We can use the above idea to prove \eqref{eq22} using  induction. \qed \end{proof}


Finally,  we  consider the vertices from $V_{\overline{T_{t}}}$ which are further partitioned into $V_{\overline{T_{t}^{\geq .5}}}$ and $V_{\overline{{T_{t}^{ < .5}}}}$. Each vertex $i \in V_{\overline{T_{t}^{\geq .5}}}$ contributes in terms of $\gamma_i$ for $0 \geq \gamma_i \geq \frac{1}{2}$ and each vertex $i \in {V_{\overline{T_{t}^{< .5}}}}$ contributes in terms of $\gamma_i$ for $\frac{-1}{2} \leq \gamma_i \leq 0$. 


\begin{lemma}\label{lemma_last} If ${H_t}$ has $Q$ connected components then the following relationship holds among the vertices from $V_{\overline{T_{t}}}$: 
\begin{equation}\label{eq25}
\sum\limits_{{i \in V_{\overline{T_{t}^{<.5}}}}}{\gamma_i} + \sum\limits_{{i \in V_{\overline{T_{t}^{\geq.5}}}}}{\gamma_i}  \geq \frac{-Q}{2}
\end{equation}
\end{lemma}
\begin{proof}
Let $E_{\overline{T_{t}^{A}}}$ be the set of edges with endpoints in the set $V_{\overline{T_{t}^{\geq .5}}}$ or one end is in the set $V_{\overline{T_{t}^{\geq .5}}}$ and the other in $V_{{T_{t}}}=V_{{T_{t}^{< .5}}} \cup V_{{T_{t}^{\geq .5}}}$. Let $E_{\overline{T_{t}^{B}}}$ be the set of edges whose one end vertex is in the set $V_{\overline{T_{t}^{< .5}}}$  and other end vertex is either in the set $V_{\overline{T_{t}^{\geq .5}}}$ or  $V_{{T_{t}^{\geq .5}}}$.
Then,  the edge set $E_{\overline{T_{t}}}$ is partitioned into $E_{\overline{T_{t}^{A}}}$ and $E_{\overline{T_{t}^{B}}}$.

From $T$,  if we delete all the edges from the set $E_{\overline{T_{t}^{B}}}$, we get a set of non-empty connected components $S$ over vertex set $V_{T_{t}} \cup V_{\overline{T_{t}^{\geq .5}}}$ and all the vertices from set $V_{\overline{T_{t}^{< .5}}}$ are now isolated vertices. Let $S'$ be the set of all those connected components in $S$ which are formed over vertices from a subset of $V_{\overline{T_{t}^{\geq .5}}}$. Let $S''$ be the set of those connected components in $S$ that contain at least one vertex from the set $V_{T_{t}}$. Note that some connected components in $S''$ may also contain some vertices from the set $V_{\overline{T_{t}^{\geq .5}}}$. Then,  $S$ may be partitioned into $S'$ and $S''$. Recall that $(V_{T_{t}}, E_{T_{t}} )$ has $Q$ connected components. $(V_{T_{t}}, E_{T_{t}} )$ is a subgraph of $S''$. In addition to $E_{T_{t}}$,  $S''$ may contain some edges from $E_{\overline{T_{t}^{A}}}$ whose one end vertex is in $V_{T_{t}}$ and other end vertex is in $V_{\overline{T_{t}^{\geq .5}}}$. Note that such edges either merge connected components from $(V_{T_{t}}, E_{T_{t}})$ or extend some connected component from $(V_{T_{t}}, E_{T_{t}})$ by adding some edges from $E_{\overline{T_{t}^{A}}}$ to form connected components in $S''$. In both cases,  the number of connected components doesn't increase but may decrease. Therefore,  $|S''|\leq Q$. Degree of each vertex $i \in V_{\overline{T_{t}^{< .5}}}$ is at least two. Therefore,  the number of incident edges on the vertices from set $V_{\overline{T_{t}^{< .5}}}$ is at least $2\cdot|V_{\overline{T_{t}^{< .5}}}|$ i.e. $|E_{\overline{T_{t}^{B}}}| \geq 2\cdot|V_{\overline{T_{t}^{< .5}}}|$.  

Let $T'=(V_{T'}, E_{T'})$ be a tree obtained from $T$ when all edges of each component in $S$ have been contracted. Note that in this process,  all edges of a component in $S$ form a vertex after contraction. Let $U$ be the set of all such vertices formed after contracting the edges of the components,  i.e,  and $|U| = |S|$. 
Then,  $V_{T'} = V_{\overline{T_{t}^{< .5}}} \bigcup U$ and $E_{T'} = E_{\overline{T_{t}^{B}}}$.  In tree $T'$,  $|V_{T'}| = |E_{T'}| + 1 \geq 2|V_{\overline{T_{t}^{< .5}}}| + 1 $. It implies that  $|V_{\overline{T_{t}^{< .5}}}| + |U| \geq 2|V_{\overline{T_{t}^{< .5}}}| + 1$ or we can rewrite     $|V_{\overline{T_{t}^{< .5}}}| + |S| \geq 2|V_{\overline{T_{t}^{< .5}}}| + 1$. Hence,  $|S| \geq |V_{\overline{T_{t}^{< .5}}}| + 1$.



A matching in a graph is a subset of edges with no common end vertex. A graph is bipartite if its vertex set can be split into two sets such that every edge has one end vertex in each set. In a bipartite graph, a matching $M$ saturates a vertex partition $V_1$ if for each vertex in $V_1$, there is an edge in $M$ that is incident on that vertex. 
The neighborhood of a vertex set $X$ is defined as $N(X) = \{j | (i,j) \in E \text{ and } i \in X\}$. Hall's theorem \cite{Hall} states that a bipartite graph has a matching that saturates $V_1$ if and only if $|N(X')| \geq |X'|$ for all $X' \subseteq X$.

Tree $T'=(V_{T'}, E_{T'})$ is a bipartite graph $(V_{\overline{T_{t}^{< .5}}}\cup U, E_{\overline{T_{t}^{B}}} )$ with vertex partitions $V_{\overline{T_{t}^{< .5}}}$ and $S$. Degree of each vertex $i \in V_{\overline{T_{t}^{< .5}}}$ is at least $2$ and there is no cycle in $T'$. Therefore,  $|N(X')| \geq |X'|$ for all $X' \subseteq V_{\overline{T_{t}^{< .5}}}$. Hence,  by Hall's theorem,  $T'$ has a matching $M$ that saturates $V_{\overline{T_{t}^{< .5}}}$. We have a matching, denoted by $M$, which connects the vertices from $V_{\overline{T_t^{<0.5}}}$ to distinct vertices in $U$, using edges present in $M$. There are two kinds of vertices in $U$,  one due to contracting the edges in the components present in $S'$ and the other due to contracting the edges in the components present in $S''$. Next,  we analyze the edges in matching $M$. First,  we consider the edges that link the vertices from $V_{\overline{T_{t} ^{\geq .5}}}$ with some vertices in $U$ corresponding to components in $S'$,  Afterwards,  we consider the edges that link vertices from $V_{\overline{T_{t} ^{\geq .5}}}$ with some vertices in $U$ corresponding to components in $S''$.

Recall,  each vertex $i \in V_{\overline{T_{t} ^{\geq .5}}}$ contributes  $0 \geq \gamma_i \geq \frac{1}{2}$ and each vertex $i \in {V_{\overline{T_{t}^{< .5}}}}$ contributes  $\frac{-1}{2} \geq \gamma_i \geq 0$. Given the first constraint of the LP,  $x_i^*+x_j^*\geq 1$ for each edge. Therefore,  for each edge $(i, j)$,  where $i \in {V_{\overline{T_{t}^{< .5}}}}$ and $j \in  V_{\overline{T_{t} ^{\geq .5}}}$, $\gamma_i + \gamma_j \geq 0$. For each edge $(i, j) \in M$,  where $i \in V_{\overline{T_{t}^{< .5}}}$ and  $j \in U$ corresponding to some component $S'$,  i.e. $j\in V_{S'}\subseteq V_{\overline{T_{t}^{\geq .5}}}$,  we have $\gamma_i + \gamma_j \geq 0$. 

For each edge $(i, j) \in M$,  where $i \in V_{\overline{T_{t}^{< .5}}}$ and  $j \in U$ corresponding to some component $S''$,  we have $\gamma_i \geq \frac{-1}{2}$. Recall that $|S''| \leq |Q|$,  and there are at most $S''$ many edges present in the matching that connect some vertices from $V_{\overline{T_{t}^{< .5}}}$ with some vertices in $U$ corresponding to components in $S''$. Therefore,   we have
 


 \begin{equation*}
\sum\limits_{{i \in V_{\overline{T_{t}^{<.5}}}}}{\gamma_i} + \sum\limits_{{i \in V_{\overline{T_{t}^{\geq.5}}}}}{\gamma_i}  \geq \frac{-Q}{2}
\end{equation*}
\qed
\end{proof}
Now we are ready to prove Lemma~\ref{lemma6.6}.
\begin{proof}[Proof of Lemma~\ref{lemma6.6}:]

After adding  \eqref{eq10} and \eqref{eq22} we get the following equation:

\begin{multline}\label{eq15}
\sum\limits_{i \in V_{D_{t}}}{x_i^*} - \sum\limits_{(i, j) \in E_{D_{t}}}{y_{ij}^*}(1-w_{ij}) + \sum\limits_{i \in V_{T_{t}}}{\gamma_i} -\sum\limits_{(i, j) \in E_{T_{t}^ {\geq.5}}}{\beta_{i, j}}(1-w_{ij}) \\ +\sum\limits_{(i, j) \in E_{T_{t}^ {<.5}}}{\zeta_{i, j}(1-w_{ij})}
- \sum\limits_{(i, j) \in E_{NT_{t}}}{y^*_{i, j}(1-w_{ij})} + \sum\limits_{(i, j) \in E_{\overline{T_{t }^{< .5}}}}{\zeta_{ij}}(1-w_{ij}) \\
- \sum\limits_{(i, j) \in E_{\overline{T_{t }^ {\geq .5}}}}{\beta_{ij}}(1-w_{ij}) 
- \sum\limits_{(i, j) \in E_{\overline{NT_{t }}}}{y_{ij}^*}(1-w_{ij}) \geq \frac{Q}{2}.
\end{multline}

In the sub-graph $\overline{H_{t}}$,  $\sum\limits_{i \in V_{\overline{D_{t}}}}{x_i^*} - \sum\limits_{(i, j) \in E_{\overline{D_{t}}}}{y_{ij}^*} \geq 0$ and $w_{ij}y_{ij}^* \geq 0$ for all $(i, j) \in E_{\overline{D_{t}}}$. Hence,  we have $\sum\limits_{i \in V_{\overline{D_{t}}}}{x_i^*} - \sum\limits_{(i, j) \in E_{\overline{D_{t}}}}{(1-w_{ij})y_{ij}^*} \geq 0$. Therefore,  given $V_{\overline{D_{t}}}$ and edges from $E_{\overline{D_{t}}}$,  \eqref{eq15} is rewritten as
\begin{multline}\label{eq24}
\sum\limits_{i \in V_{D_{t}}}{x_i^*} - \sum\limits_{(i, j) \in E_{D_{t}}}{y_{ij}^*}(1-w_{ij}) + \sum\limits_{i \in V_{T_{t}}}{\gamma_i} -\sum\limits_{(i, j) \in E_{T_{t}^ {\geq.5}}}{\beta_{i, j}}(1-w_{ij}) \\+\sum\limits_{(i, j) \in E_{T_{t}^{ <.5}}}{\zeta_{i, j}(1-w_{ij})} 
- \sum\limits_{(i, j) \in E_{NT_{t}}}{y^*_{i, j}(1-w_{ij})} + \sum\limits_{(i, j) \in E_{\overline{T_{t }^{< .5}}}}{\zeta_{ij}}(1-w_{ij}) \\
- \sum\limits_{(i, j) \in E_{\overline{T_{t }^{\geq .5}}}}{\beta_{ij}}(1-w_{ij}) 
- \sum\limits_{(i, j) \in E_{\overline{NT_{t }}}}{y_{ij}^*}(1-w_{ij}) + \sum\limits_{i \in V_{\overline{D_{t}}}}{x_i^*} - \sum\limits_{(i, j) \in E_{\overline{D_{t}}}}{y_{ij}^*}(1-w_{ij}) \geq \frac{Q}{2}
\end{multline}

After adding \eqref{eq24} and \eqref{eq25} we get,

\begin{multline}\label{eq26}
\sum\limits_{i \in V_{D_{t}}}{x_i^*}-\sum\limits_{(i, j) \in E_{D_{t}}}{y_{ij}^*}(1-w_{ij}) + \sum\limits_{i \in V_{T_{t}}}{\gamma_i} -\sum\limits_{(i, j) \in E_{T_{t}^ {\geq.5}}}{\beta_{ij}}(1-w_{ij}) +\sum\limits_{(i, j) \in E_{T_{t}^ {<.5}}}{\zeta_{ij}(1-w_{ij})} \\
- \sum\limits_{(i, j) \in E_{NT_{t}}}{y^*_{ij}(1-w_{ij})} + \sum\limits_{(i, j) \in E_{\overline{T_{t }^{< .5}}}}{\zeta_{ij}}(1-w_{ij}) 
- \sum\limits_{(i, j) \in E_{\overline{T_{t }^{\geq .5}}}}{\beta_{ij}}(1-w_{ij}) \\
- \sum\limits_{(i, j) \in E_{\overline{NT_{t }}}}{y_{ij}^*}(1-w_{ij}) + \sum\limits_{i \in V_{\overline{D_{t}}}}{x_i^*} - \sum\limits_{(i, j) \in E_{\overline{D_{t}}}}{y_{ij}^*}(1-w_{ij}) +
 \sum\limits_{{i \in V_{\overline{T_{t}^{<.5}}}}}\gamma_i + 
\sum\limits_{i \in V_{\overline{T_{t}^{\geq.5}}}}\gamma_i
\geq 0
\end{multline}

\begin{figure}[ht]
    \centering
    \includegraphics[width=0.4\textwidth]{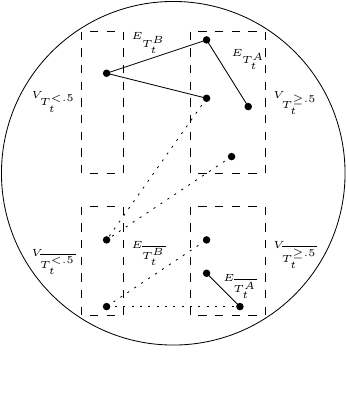}
    \caption{Partition of vertices and edges}
    \label{fig:my_label}
\end{figure}

To \eqref{eq26}, we add terms of sub-graph $H_{t}$ and $H_{\overline{t}}$, i.e,  
$V_{D}=V_{{D_{t}}} \cup  V_{\overline{D_{t}}}$,  $E_{D}=E_{{D_{t}}} \cup  E_{\overline{D_{t}}}$,  $V_{T}=V_{{T_{t}}} \cup  V_{\overline{T_{t}}}$, 
$E_{T^ {\geq.5}}=E_{{T_{t}^ {\geq.5}}} \cup  E_{\overline{T_{t}^ {\geq.5}}}$,  $E_{T^ {<.5}}=E_{{T_{t}^ {<.5}}} \cup  E_{\overline{T_{t}^ {<.5}}}$,  and $E_{NT}=E_{{NT_{t}}} \cup  E_{\overline{NT_{t}}}$.  Then, the following inequality holds in a connected component $C$.

\begin{multline}\label{eq27}
   \sum\limits_{i \in V_{D}}{x_i^*} - \sum\limits_{(i, j) \in E_{D}}{y_{ij}^*}(1-w_{ij}) + \sum\limits_{i \in V_{T}}{\gamma_i} -\sum\limits_{(i, j) \in E_{T^{\geq.5}}}{\beta_{i, j}}(1-w_{i, j}) \\ +\sum\limits_{(i, j) \in E_{T^ {<.5}}}{\zeta_{i, j}}(1-w_{i, j}) -\sum\limits_{(i, j) \in E_{NT}}{y^*_{i, j}}(1-w_{i, j}) \geq 0.
\end{multline}
This proves Lemma~\ref{lemma6.6}.\qed
\end{proof}

Finally, we prove Lemma~\ref{Th_main} based on the lemmas and results derived in this section.

\begin{proof}[Proof of Lemma~\ref{Th_main}]

After adding $\frac{1}{2}\sum\limits_{V_{T}}{x_i'}$ term, $-\frac{1}{2}\sum\limits_{E_{T^{\geq .5}}}{y_{ij}'}(1-w_{ij})$ term, and $-\frac{1}{2}\sum\limits_{E_{T^{
<.5}}}{y_{ij}'}(1-w_{ij})$ term to both sides of \eqref{eq8} of Lemma~\ref{lemma6.6}, we get

\begin{multline*}
   \sum\limits_{i \in V_{D}}{x_i^*} - \sum\limits_{(i, j) \in E_{D}}{y_{ij}^*}(1-w_{ij}) + \frac{1}{2}\sum\limits_{V_{T}}{x_i'} + \sum\limits_{i \in V_{T}}{\gamma_i} -\frac{1}{2}\sum\limits_{E_{T^{\geq .5}}}{y_{ij}'}(1-w_{ij}) -\sum\limits_{(i, j) \in E_{T^{\geq.5}}}{\beta_{i, j}}(1-w_{i, j}) \\ 
   -\frac{1}{2}\sum\limits_{E_{T^{
   <.5}}}{y_{ij}'}(1-w_{ij})
   +\sum\limits_{(i, j) \in E_{T^ {<.5}}}{\zeta_{i, j}}(1-w_{i, j}) -\sum\limits_{(i, j) \in E_{NT}}{y^*_{i, j}}(1-w_{i, j}) \\
   \geq \frac{1}{2}\sum\limits_{V_{T}}{x_i'} - \frac{1}{2}\sum\limits_{E_{T^{\geq .5}}}{y_{ij}'}(1-w_{ij}) - \frac{1}{2}\sum\limits_{E_{T^{
   <.5}}}{y_{ij}'}(1-w_{ij})
\end{multline*}

$\implies \sum\limits_{V_{D}}{x_i^*} + \sum\limits_{V_{T}}{x_i^*} - \sum\limits_{V_{D}}{y_{ij}^*}(1-w_{ij}) - \sum\limits_{E_{T ^ {\geq .5}}}{y_{ij}^*}(1-w_{ij}) - \sum\limits_{E{T ^ {< .5}}}{y_{ij}^*}(1-w_{ij}) -\sum\limits_{(i, j) \in E_{NT}}{y^*_{i, j}}(1-w_{i, j}) \geq \frac{1}{2}\sum\limits_{V_{T}}{x_i'} - \frac{1}{2}\sum\limits_{E_{T^{\geq .5}}}{y_{ij}'}(1-w_{ij}) - \frac{1}{2}\sum\limits_{E_{T^{
   <.5}}}{y_{ij}'}(1-w_{ij})$

$\implies \sum\limits_{V_{C}}{x_i^*} - \sum\limits_{E_{C}}{y_{ij}^*}(1 - w_{i, j}) \geq \frac{1}{2}\bigl[\sum\limits_{V_{T}}{x'_i} - \sum\limits_{E_{T}}{y'_{i, j}}(1 - w_{i, j}) \bigr] $\qed

\end{proof}

\section{Improved Approximation for Bounded Forest Cover (BFC)}
\textcolor{black}{\textit{Bounded Forest Cover} (BFC) was studied in \cite{GPS23}, where they show that the problem is NP-complete 
and gave an $8$-approximation algorithm.} We improve this factor to $6$ using the $2$-factor approximation algorithm for the forest cover problem given in Sectio~\ref{sec:det}. We define the bounded forest cover problem next.

\noindent{\sc Bounded Forest Cover(BFC):} We are given a graph $G$ with positive weights on the edges and a parameter $\lambda \ge 0.$ The goal is to find a minimum-sized collection of trees $T_1, T_2, \ldots T_k$ such that the weight of each $T_i \le \lambda$ and the vertices in $\cup_{i=1}^k T_i$ form a vertex cover of the graph. By minimum-sized collection, we mean $k$ should be the smallest possible. We can assume that the weight on any edge is at most $\lambda.$

We will construct a new graph $G'$ from $G$ with weights on the edges in the interval $[0,1]$ as follows: the vertex and edge set remains unchanged, if $w_e$ is the weight on edge $e \in G$ then $w'_e$ the weight on $e \in G'$ is defined as follows:
$$w'_e = 
        \begin{cases} 
            1, & if w_e > \lambda/2 \\
            2 w_e/ \lambda, & otherwise.
        \end{cases} 
$$
\textcolor{black}{
\begin{lemma}[Lemma 2 in \cite{KhaniS14}] \label{lemma:2}
Given a tree $T$ with weight $w(T)$ and a parameter $\beta >0$ such that all the edges of $T$
have weight at most $\beta$, we can edge-decompose $T$ into trees $T_1, \ldots T_k$ with $k \le \max \{ w(T)/\beta, 1\}$
such that $w(T_i) \le 2\beta$ for each $1 \le i \le k.$
\end{lemma}
Next, we solve the forest cover problem approximately using the results in the previous section. This approximate solution gives us a collection $C$ of $k$ trees in $G'$. Each tree $T_i$ in the collection $C$ is edge-decomposed into a collection $S_i$ of tree $T_1, T_2, \ldots , T_m$ using Lemma \ref{lemma:2} (with $\beta=1$). Each tree in collection $S_i$ has weight at most $2$ in $G'.$}

\begin{observ} \label{obsv:1}
\textcolor{black}{ Let $OPT'$ be the value of the optimal solution to the forest cover problem. Let the optimal solution (trees; each with weight at most $\lambda$) to BFC be $\Lambda_1, \ldots, \Lambda_{OPT}$. Then, the value of the optimal solution to bounded forest cover is $OPT.$ Since the optimal solution to BFC is a feasible solution the forest cover problem and $ w'(\Lambda_i) \le 2$ (feasible solution to BFC), we get $OPT' \le \sum_{i=1}^{OPT} w'(\Lambda_i) + OPT \le 3 OPT$.}
\end{observ}

Next, we need to convert the solution in $G'$ to a solution to the bounded forest cover problem and analyze the performance ratio. All the vertices that are in the solution to the 2-approximate solution to the forest cover problem forms a vertex cover in the bounded forest cover problem. The trees in $S_i$ obtained using the edge-decomposition in Lemma \ref{lemma:2} (for all $i$) are the trees in the approximate solution to the BFC. 

\begin{observ}\label{obsv:2}
If any of trees use a weight of edge $1$, then we simply remove the edge from the solution; this does not change the objective function value of the forest cover problem. Therefore, without loss of generality, we assume that all the trees use only edges with weights $<1$ in $G'$. Equivalently, in $G$, the weight of the edges of all the trees is at most $\lambda/2.$ 
\end{observ}
The number of trees in BFC, by Lemma~\ref{lemma:2} are $\sum_{i=1}^k\max\{w'(T_i), 1\} \le \sum_{i=1}^k ($ $w'(T_i)+ 1)$
where each $T_i$ is a connected component in some $G_0$ (the graph with weights 0/1 used in the previous section) and $w'(T_i)=\sum_{e\in T_i} w'_e$ and $k$ is the number of connected connected components. $\sum_{i=1}^k (w'(T_i)+ 1)$ is by definition the value of the 2-approximate solution to the forest cover problem. Therefore, using Observation \ref{obsv:1}, we get $\sum_{i=1}^k (w'(T_i)+ 1) \le 2 OPT' \le 6 OPT $.
Therefore, the number of trees is at most a multiplicative factor of $6$ of the minimum possible.  Each tree has weight $w(T) = \sum_{e\in T} w_e = \sum_{e\in T} w'_e \lambda/2$ (by Observation \ref{obsv:2}). Also, $\sum_{e\in T} w'_e \le 2$, therefore  $w(T) \le \lambda$. Therefore, we have the following theorem.

\begin{theorem}\label{th_BFC}
There exists a $6$-factor approximation algorithm for the bounded forest cover problem.
\end{theorem}

\section{Conclusions}
We suggest an improved method to approximate the bounded component forest cover problem using LP-rounding. This method is based on the 2-approximation algorithm for the forest cover problem also developed here. We also give a probabilistic approximation algorithm for forest cover with an approximation factor of $2+\epsilon$. The probabilistic technique may be of independent interest. It would be valuable to explore ways to improve the approximation factor and consider covering graphs with other types of subgraphs. Additionally, we could investigate bidirectional linear programming formulations for the forest cover problem.

\bibliographystyle{splncs04}
\bibliography{references}

\end{document}